\documentclass[11pt]{article}

\usepackage[margin=1in]{geometry}
\usepackage{amsmath,amssymb,amsthm,mathtools}
\usepackage{booktabs}
\usepackage{enumerate}
\usepackage{microtype}
\usepackage{authblk}
\usepackage{aliascnt}
\usepackage{algorithm}
\usepackage[noend]{algpseudocode}
\usepackage[numbers,sort&compress]{natbib}
\usepackage[hidelinks]{hyperref}
\usepackage[nameinlink,capitalize,noabbrev]{cleveref}
\usepackage{xcolor}

\makeatletter
\renewcommand{\theHALG@line}{\thealgorithm.\arabic{ALG@line}}
\makeatother

\newtheorem{theorem}{Theorem}[section]

\newaliascnt{lemma}{theorem}
\newtheorem{lemma}[lemma]{Lemma}
\aliascntresetthe{lemma}

\newaliascnt{proposition}{theorem}
\newtheorem{proposition}[proposition]{Proposition}
\aliascntresetthe{proposition}

\newaliascnt{corollary}{theorem}
\newtheorem{corollary}[corollary]{Corollary}
\aliascntresetthe{corollary}

\theoremstyle{definition}
\newaliascnt{definition}{theorem}
\newtheorem{definition}[definition]{Definition}
\aliascntresetthe{definition}

\newaliascnt{remark}{theorem}

\aliascntresetthe{remark}

\crefname{theorem}{Theorem}{Theorems}
\crefname{lemma}{Lemma}{Lemmas}
\crefname{proposition}{Proposition}{Propositions}
\crefname{corollary}{Corollary}{Corollaries}
\crefname{definition}{Definition}{Definitions}
\crefname{remark}{Remark}{Remarks}
\crefname{algorithm}{Algorithm}{Algorithms}

\newcommand{\N}{N}
\newcommand{\M}{M}
\newcommand{\val}{v}
\newcommand{\wt}{w}
\newcommand{\R}{\mathbb{R}}
\newcommand{\Bundles}{\mathcal{B}}
\newcommand{\Chores}{\mathcal{Z}}

\newcommand{\SW}{\operatorname{SW}}
\newcommand{\OPT}{\operatorname{OPT}}
\newcommand{\OPTWEF}{\operatorname{OPT}_{\mathrm{WEF1}}}
\newcommand{\WMMS}{\mathsf{WMMS}}
\newcommand{\WMMC}{\operatorname{WMMS}^{\mathrm{cost}}}
\newcommand{\argmax}{\mathop{\mathrm{arg\,max}}}

\title{\bf Weighted Fair Division of Indivisible Mixed Manna}
\author[]{Nicholas Teh}
\affil[]{University of Oxford, UK}
\date{\empty}

\begin{document}
\maketitle

\begin{abstract}
    We study weighted fair division of indivisible mixed manna under additive valuations. First, we resolve the general existence open question for weighted envy-freeness up to one item (WEF1), and show that every instance with arbitrary positive entitlements admits a complete WEF1 allocation computable in polynomial time. We then show that existence does not imply any welfare guarantee, i.e., the utilitarian price of WEF1 is infinite, even for two unweighted agents with normalized valuations, common item signs, and singleton values in a fixed four-value set; a welfare-maximizing WEF1 allocation in the construction is fractionally Pareto optimal.
    
      Second, suppose each agent $i$ has a number $a_i>0$ such that their valuation for any item is $-a_i$, $0$, or $a_i$. Then, for arbitrary entitlements, a weighted maximin share (WMMS) allocation always exists, is computable in polynomial time, and can be chosen to be fractionally Pareto optimal. An exact formula for each WMMS value leads to a polynomial-time flow algorithm. In this class, every WEF1 allocation satisfies a best possible additive WMMS guarantee whose loss depends on the agent's entitlement relative to the largest entitlement. Thus maximum entitlement agents receive exact WMMS and, under equal entitlements, every WEF1 allocation is also MMS-fair. Allowing a second positive magnitude can violate exact WMMS, while unrestricted entitlement ratios rule out any fixed multiplicative WMMS guarantee compatible with WEF1 for chores.
  \end{abstract}

\section{Introduction}\label{sec:introduction}

Consider an online platform or computing system that must allocate advertising slots, dataset licenses, servers, and data processing or ML tasks among users or teams with different priorities. How should it do so fairly? Some items are desirable, some tasks are burdensome, and the same item may be desirable to one user but undesirable to another. Moreover, the items and tasks are indivisible, so an allocation that exactly reflects both preferences and priorities may not exist. Similar questions arise when housemates who pay different shares of the rent divide bedrooms, parking spaces, and furniture and assign household tasks: a parking space or piece of furniture may be useful to one housemate but unwanted by another, while cooking may be enjoyable to one but burdensome to another.

The study of how resources and responsibilities should be divided among interested agents is commonly known as \emph{fair division}. When all items are desirable, they are called \emph{goods}; when all items are burdensome, they are called \emph{chores}. Many applications contain both types, and an item may even have positive value for one agent and negative value for another. Such a collection is called a \emph{mixed manna}; see \citet{LiuEtAl24} for a recent survey. Indivisibility makes fairness nontrivial to achieve even in the goods setting: for example, two agents who both value a single good positively cannot receive an \emph{envy-free} allocation.\footnote{An allocation is \emph{envy-free} if every agent values her own bundle at least as highly as every other bundle.}

A further issue is that agents may have unequal \emph{entitlements}. Entitlements can represent, for example, different rent contributions, group sizes, or prior claims. In the unweighted setting, all agents have the same entitlement. In the weighted setting, each agent $i$ has a positive entitlement $w_i$, and fairness should take these entitlements into account. Weighted fair division of indivisible items has therefore received increasing attention; see \citet{Suksompong25} for a recent review.

Two broad approaches to fairness are especially prominent. The first is based on \emph{envy}. An allocation is \emph{weighted envy-free} if every agent values her own bundle per unit entitlement at least as highly as every other bundle per unit entitlement. Exact weighted envy-freeness need not exist for indivisible items, which motivates \emph{weighted envy-freeness up to one item} (WEF1). For goods, a weighted envy may be resolved by removing one positively valued item from the recipient's bundle; for chores, it may be resolved by removing one negatively valued item from the observer's own bundle. In a mixed manna, either removal may be needed.

The second approach compares an agent's allocation with a \emph{share} that she can define for herself. The maximin share (MMS), introduced by \citet{Budish11} for equal entitlements, asks an agent, when there are $n$ agents, to partition the items into $n$ bundles and then evaluates the least valuable bundle in her best partition. The \emph{weighted maximin share} (WMMS) of \citet{FarhadiEtAl19} labels the bundles by the agents, compares their values per unit entitlement, and scales the resulting minimum by the agent's own entitlement. Thus WEF1 compares pairs of allocated bundles, whereas WMMS compares an agent's allocated bundle with a partition benchmark derived from all items.

Besides fairness, we consider welfare because requiring fairness may assign some items to agents who value them less and thereby reduce the total value of the allocation. We focus on \emph{(utilitarian) social welfare}, a standard measure in fair division, defined as the sum of the agents' values for their assigned bundles~\citep{BeiEtAl19,BarmanBhaskarShah20}. Whenever the maximum welfare among WEF1 allocations is positive, the utilitarian price of WEF1 is the ratio between the maximum welfare without a fairness requirement and this maximum WEF1 welfare.

Several questions remain open.
Complete WEF1 allocations are known to exist for indivisible goods and indivisible chores~\citep{CISZ21,WZZ25}. In the unweighted setting, complete EF1 allocations are known to exist for mixed manna~\citep{BhaskarSricharanVaish21,ACIW22}. For weighted mixed manna, however, complete WEF1 existence was left open even for three agents~\citep{GS24}.
On the share-based side, exact WMMS is known for binary additive goods and binary chores~\citep{ChenEtAl26,AzizChanLi19}. The goods result allows only nonnegative values, whereas the chores result allows only nonpositive values; neither covers mixed manna. For unrestricted mixed manna, no positive uniform multiplicative MMS guarantee exists even with equal entitlements~\citep{KulkarniMehtaTaki20}.

We work throughout with additive valuations and arbitrary positive entitlements. We address three questions. First, does every weighted mixed manna instance admit a complete WEF1 allocation, i.e., that every item is assigned, and can such an allocation be found in polynomial time? Second, how much utilitarian welfare can be lost by requiring WEF1? Third, under what assumptions can every agent receive her exact WMMS, and what WMMS guarantee follows from WEF1? Our answers distinguish the general additive setting from a valuation class in which each agent assigns the same absolute value to every item that she values nonzero.

\subsection{Our Results}\label{subsec:our-results}

We state our main results and briefly highlight the main ideas behind their proofs.

\paragraph{1. A complete WEF1 allocation always exists and can be computed in polynomial time.}
We prove that every additive instance of indivisible mixed manna with arbitrary positive entitlements admits a complete WEF1 allocation, and that such an allocation can be computed in polynomial time (\Cref{thm:main}). This resolves the existence question left open by \citet{GS24}. 

The proof starts with an initial bundling step adapted from
\citet[Algorithm~1 and Lemma~4.5]{AzizEtAl25Cake}.
\citet[Section~3 and Proposition~3.1]{AzizEtAl26BoBW}
use the same bundling step, while
\citet[Section~4.1 and Lemma~4.2]{LuMackenzieSuzuki26}
give a related variant with an additional unit upper bound for interested agents.
In our terminology, their meta-goods and interest sets are called acceptable bundles
and acceptance sets. The resulting acceptance sets are pairwise disjoint, and adding
any remaining objective chore to any acceptable bundle makes the combined bundle
negative for every agent.

After this preprocessing, our allocation arguments use weighted picking sequences.
The earlier mixed-manna papers also distinguish constructions according to the number
of remaining objective chores, but the allocation rules in our two branches are
different. In the many-chore branch, we use the reversed weighted picking sequence
of \citet{WZZ25}, together with the continuous representation of weighted picking
sequences used by \citet{LiLiWu22}. The usual RWPS comparison deletes the observer's
last pick and compares the result with the recipient's entire chore bundle. Our
two-sided comparison also sets aside the recipient's first pick. For uniform chore costs, the comparison holds exactly when $q_i\le q_j$.

The maximum-$q$ assignment is chosen to satisfy this condition: if observer $i$ values an acceptable bundle assigned to recipient $j$ nonnegatively, then both agents
belong to its acceptance set and hence $q_i\le q_j$. The set aside chore can therefore be combined with the acceptable bundle, and the bundling property makes this combination negative for the observer. \Cref{sec:multi-reserve} extends this argument to several deleted and set aside picks and gives a WEF1 condition that does not require pairwise disjoint acceptance sets.

In the few-chore branch, the remaining acceptable bundles are allocated using the
weighted picking sequence of \citet{CISZ21}. Skipped turns are handled through
zero-valued dummy items; compare the related mixed manna argument of
\citet{GS24}. The role of this step in our proof is to obtain WEF1 at the level of
acceptable bundles while assigning every bundle only to an agent who values it
nonnegatively. The few-chore refinement then allows deletion of a whole acceptable
bundle to be replaced by deletion of one original item.

\paragraph{2. The utilitarian price of WEF1 is infinite.}
The existence of WEF1 allocations does not imply that one can choose a WEF1 allocation whose utilitarian welfare is close to the maximum possible welfare. We give a family of two-agent instances with equal entitlements such that both agents value the full set of items at one, every item has the same sign for both agents, and every singleton value belongs to $\{-2/3,-1/3,1/3,2/3\}$.
For the instance indexed by $k$, the maximum welfare is $(k+4)/3$, whereas the maximum welfare among WEF1 allocations is $5/3$. Hence the ratio is $(k+4)/5$ and tends to infinity (\Cref{thm:price}).

The example rules out several simple explanations for the welfare loss. The agents agree on whether every item is a good or a chore, their total values are normalized, and the singleton values come from a fixed four-value set. Nevertheless, WEF1 requires many items to be assigned to the agent who values them less. Moreover, a welfare-maximizing WEF1 allocation in the construction is fractionally Pareto optimal (fPO). Thus, fPO does not by itself provide an approximation to maximum utilitarian welfare. For comparison, with two unweighted agents and only chores, the EF1 price is $5/4$~\citep{WZZ25}; allowing both common goods and common chores makes the price unbounded.

\paragraph{3. WMMS for equal-magnitude mixed values.}
We then study a mixed valuation class that subsumes the special settings of binary goods and binary chores. Each agent has a positive magnitude such that every singleton value is either the negative of that magnitude, zero, or the magnitude itself. We call these \emph{equal-magnitude mixed values}. The signs and magnitudes may differ across agents.

We prove that, for this class and arbitrary positive entitlements, an exact WMMS allocation always exists and can be computed in polynomial time (\Cref{thm:wmms-main}). The allocation maximizes the sum of the agents' utilities after each agent's valuation is divided by her magnitude, even over fractional allocations, and is therefore fPO. If all agents have the same magnitude, it also maximizes ordinary utilitarian welfare.

After normalization, we derive an exact ceiling formula showing that an agent's WMMS is determined by her total normalized value and the entitlement vector. The formula solves each agent's labeled partition problem and gives an integer utility target. Individual optimal partitions do not, however, give one allocation that meets all agents' targets. We prove Hall's inequalities for every subset of agents and use a matching to assign the required $+1$ items. The same ceiling inequality then shows that all common chores can be assigned without taking any agent below her target. Assigning every item to an agent who values it most after normalization gives the welfare and fPO conclusions.

\paragraph{4. WEF1 gives a best possible additive WMMS guarantee.}
The resulting WMMS formula gives a direct relation between WEF1 and WMMS in the equal-magnitude class. For a fixed observer, the three alternatives in the definition of WEF1 give integer upper bounds on the value she assigns to every other allocated bundle. Summing these bounds and comparing the result with the ceiling formula yields an entitlement-dependent additive guarantee (\Cref{thm:wef1-wmms-mixed}). We then show that its coefficient is best possible: equality can occur for every agent whose entitlement is smaller than the largest, even when all agents have identical valuations over one common good and two common chores.

The shortfall is always smaller than her magnitude and decreases as her entitlement approaches the largest entitlement. Thus, every agent with maximum entitlement receives exact WMMS, and under equal entitlements every WEF1 allocation is MMS-fair. More generally, an agent below WMMS can reach it by adding one positively valued item from outside her bundle or removing one negatively valued item from her own bundle. If the largest entitlement is at most $\kappa$ times the smallest, then each agent's shortfall below WMMS is at most her magnitude times $1-1/\kappa$, and this coefficient is best possible.

We also show why the assumptions cannot be relaxed in two simple ways. For every $r>1$, there is a two-agent, three-item instance with singleton values in $\{-1,1,r\}$ that has no exact WMMS allocation. In addition, even for two identical unit chores, no fixed multiplicative WMMS guarantee is compatible with WEF1 when entitlement ratios are unrestricted. An additive comparison is therefore appropriate for mixed manna, where an agent's WMMS may be positive, zero, or negative.

\paragraph{Paper outline.}
\Cref{sec:model} gives the model and the fairness notions. \Cref{sec:wef1} proves the WEF1 existence theorem and the welfare lower bound. \Cref{sec:wmms} proves the exact WMMS theorem, the relation between WEF1 and WMMS, and the limits described above. \Cref{sec:conclusion} concludes. \Cref{sec:multi-reserve} gives the several-pick two-sided RWPS comparison, proves that its count condition is best possible for comparisons valid for every cost vector, and derives a WEF1 condition that does not require pairwise disjoint acceptance sets.

\subsection{Related Work}\label{subsec:related-work}

For recent surveys of mixed fair division and weighted fair division, see \citet{LiuEtAl24} and \citet{Suksompong25}, respectively. We discuss the work most closely related to our results.

\paragraph{Weighted envy-freeness for goods, chores, and mixed manna.}
For indivisible goods, \citet{CISZ21} introduced weighted versions of EF1, gave a polynomial-time weighted picking sequence for WEF1, and showed that WEF1 can be combined with Pareto optimality. For indivisible chores, \citet{SpringerHajiaghayiYami24} gave a polynomial-time algorithm for weighted envy-freeness up to one chore, called 1WEF in their paper, and \citet{WZZ25} introduced the reversed weighted picking sequence and used it to compute WEF1 allocations; \citet{Mahara26} later combined weighted EF1 with fractional Pareto optimality for general additive chores, and \citet{Teh26choresbounded} showed polynomial-time computation for the special case when a bounded number of additional copies of the original chores may be allocated. Beyond additive valuations, \citet{Montanari25} study weighted envy-freeness for submodular valuations in the goods setting.

For unweighted mixed manna, complete EF1 allocations always exist~\citep{BhaskarSricharanVaish21,ACIW22}, and EF1 can be combined with Pareto optimality for two agents~\citep{ACIW22}. The weighted case was left open. \citet{GS24} gave a polynomial-time algorithm for weighted envy-freeness up to one transfer (WEF1T), which is weaker than WEF1, for any number of agents. They also obtained a WEF1 and fPO allocation for two agents, but left complete WEF1 existence open even for three agents. \Cref{thm:main} resolves this question for any number of agents and arbitrary positive entitlements.

\paragraph{Bundling preprocessing and related mixed-manna constructions.}
Rules (B1) and (B2) in \cref{sec:bundling} are adapted from
\citet[Algorithm~1]{AzizEtAl25Cake}, and the first three conclusions of
\cref{lem:bundling} correspond to
\citet[Lemma~4.5]{AzizEtAl25Cake}.
\citet[Section~3 and Proposition~3.1]{AzizEtAl26BoBW}
use the same preprocessing.
\citet[Section~4.1 and Lemma~4.2]{LuMackenzieSuzuki26}
give a related variant under unit-value normalization, with an additional upper
bound on the value of a meta-good to an interested agent.
Our acceptable bundles and acceptance sets correspond to the meta-goods and
interest sets in these papers. This initial preprocessing is the part of our WEF1
proof that directly uses their bundling construction.

The same earlier papers contain other constructions with some similarities to the
organization of our proof. \citet{AzizEtAl25Cake,AzizEtAl26BoBW} also distinguish
cases according to the number of objective chores remaining after preprocessing.
In their many-chore constructions, acceptable meta-goods are explicitly associated
with distinct objective chores; related associations appear in the synchronized
construction of \citet[Section~4]{AzizEtAl26BoBW} and the abundant-chore construction
of \citet[Section~4.4]{LuMackenzieSuzuki26}. Our many-chore branch does not use
these pairing procedures. It first allocates every objective chore by RWPS and then
uses the two-sided comparison and the maximum-$q$ assignment.

For few remaining chores, \citet[Algorithm~4 and Lemma~4.11]{AzizEtAl25Cake}
and the good-minimality and small goods-only preprocessing of
\citet[Section~5 and Proposition~5.1]{AzizEtAl26BoBW}
contain related refinements. The ordered step in which selected agents take all
remaining meta-goods that they value nonnegatively also appears in
\citet[Section~5]{AzizEtAl26BoBW}. These similarities concern the refinement and
this ordered assignment step, rather than the complete weighted allocation argument.
Our branch combines them with a weighted picking sequence and uses
\cref{lem:refinement} to replace deletion of an acceptable bundle by deletion of
one original item.

\paragraph{Weighted picking-sequence methods.}
For indivisible goods, \citet{CISZ21} use a weighted picking sequence to compute
WEF1 allocations. \citet{LiLiWu22} give a continuous representation of weighted
picking sequences in their study of weighted proportionality for chores, and
\citet{WZZ25} use the reversed weighted picking sequence to compute WEF1
allocations of chores. Related skipped-turn and zero-valued dummy-item arguments
for mixed manna appear in \citet{GS24}. 

Our many-chore argument uses RWPS but requires a comparison different from its
usual one-sided WEF1 comparison: after deleting the observer's last actual pick,
we also set aside the recipient's first actual pick. This gives the condition used
by the maximum-$q$ assignment. Our few-chore argument uses the weighted goods
picking sequence, but applies it to acceptable bundles subject to the requirement
that every bundle be assigned to an agent who values it nonnegatively.

\paragraph{Other fairness and efficiency results for mixed manna.}
Weighted proportionality up to one item (WPROP1) compares an agent's bundle with her entitlement-weighted value for all items, rather than with every other bundle. \citet{AMS20} proved that every additive mixed instance admits a WPROP1 and fPO allocation and gave a strongly polynomial-time algorithm. This result does not settle WEF1: WPROP1 is not a pairwise envy condition, and WEF1 need not imply WPROP1 even for goods~\citep{GS24,CSHS24}.

Recent work also obtains Pareto optimality together with relaxations of envy-freeness that differ from EF1. \citet{BarmanEtAl25} proved existence of an EFR-$(n-1)$ and PO allocation, where a common set of at most $n-1$ items may be reassigned when checking each agent. \citet{BarmanVerma26} proved existence of an IEF1 and PO allocation, where each agent may add one item to or remove one item from her own bundle when checking envy. Neither result gives an EF1 allocation. 
For two agents with strictly increasing valuations over goods, \citet{Teh26EF1PO} shows that every instance with at most seven items admits an EF1 and PO allocation, whereas an eight-item instance need not. 
The existence of an EF1 and PO allocation remains open even for unweighted mixed manna~\citep{AzizEtAl26BoBW}. In contrast, \citet[Theorem~4.3]{MackenzieSuzuki26} give a three-agent, four-item additive mixed manna instance with no allocation that is both EF1 and fPO, although the instance admits an EF1 and PO allocation. Since equal entitlements are a special case of the weighted setting, WEF1 and fPO cannot always be combined. Our WEF1 theorem answers the existence question without an efficiency requirement.

In the temporal mixed manna setting, \citet{Choi2026temporalMM} give an online rule that maintains EF1 and PO after every item arrival when the agents' valuations agree after agent-specific scaling and the scaling factors are known in advance.

\paragraph{The price of fairness.}
For normalized indivisible goods, the utilitarian price of EF1 is $\Theta(\sqrt{n})$~\citep{BeiEtAl19,BarmanBhaskarShah20}. For indivisible chores, the EF1 price is $5/4$ for two unweighted agents and is unbounded from three agents onward~\citep{WZZ25}. Our lower bound differs in two ways: it is unbounded with only two agents, and it remains unbounded when the best WEF1 allocation is fractionally Pareto optimal. The example also uses common item signs and a fixed four-value set.

\paragraph{Maximin shares and unequal entitlements.}
The maximin share was introduced by \citet{Budish11} for equal entitlements. \citet{FarhadiEtAl19} introduced WMMS for indivisible goods with unequal entitlements and showed that the best general multiplicative guarantee is $1/n$. Exact WMMS is known for binary additive goods~\citep{ChenEtAl26} and binary chores~\citep{AzizChanLi19}. These results treat only one sign of value at a time.

For unrestricted mixed manna, no positive uniform multiplicative MMS guarantee exists even with equal entitlements~\citep{KulkarniMehtaTaki20}. \citet{Hsu24} proved exact MMS existence for mixed instances with at most $n+5$ items under additional assumptions; the entitlements are equal and the number of items is bounded. Work on ternary mixed values has also obtained envy-based fairness results in the unweighted setting~\citep{AleksandrovWalsh20}, but it does not give WMMS guarantees for unequal entitlements. To the best of our knowledge, \Cref{thm:wmms-main} is the first exact WMMS existence theorem for mixed manna that allows both arbitrary positive entitlements and an arbitrary number of items.

The general WPROP1 and fPO theorem of \citet{AMS20} does not imply this result. WPROP1 uses the entitlement-weighted value of all items, whereas WMMS is based on the least value per unit entitlement in the agent's best labeled partition. Even for nonnegative valuations, WPROP1 does not imply any positive approximation of WMMS~\citep[Corollary~4.21]{CSHS24}.

\paragraph{Relations between WEF1 and WMMS.}
For general goods instances with unequal entitlements, WEF1 does not imply any positive multiplicative approximation of WMMS~\citep[Proposition~6.2]{CISZ21}. Binary goods behave differently: \citet[Lemma~61]{GS26Relations} showed that WEF1 implies WMMS for arbitrary entitlements. Their results also imply that, with equal entitlements, EF1 entails MMS for additive values in $\{-1,0,1\}$. \Cref{thm:wef1-wmms-mixed} determines what remains true when positive and negative values coexist and entitlements may differ. It gives exact WMMS to agents with maximum entitlement, exact MMS under equal entitlements, and a best-possible entitlement-dependent additive guarantee for every other agent.

\section{Model and Fairness Notions}\label{sec:model}

An \emph{instance} $\mathcal{I}$ consists of a set of \emph{agents} $\N=[n]$, a set $\M$ of $m$ indivisible \emph{items}, positive \emph{entitlements} $(\wt_i)_{i\in\N}$, and additive \emph{valuations} $(\val_i)_{i\in\N}$ with $\val_i:2^\M\to\R$.\footnote{For the computational statements, all valuations and entitlements are rational numbers encoded in binary.}
Without loss of generality, we assume $\sum_{i\in\N}\wt_i=1$.
For notational simplicity, we denote $\val_i(o)$ for $\val_i(\{o\})$. 
An allocation $A=(A_1,\dots,A_n)$ is a partition of $\M$.
We refer to the case where all agents have equal entitlements as the \emph{unweighted} setting.

We consider the setting where $\M$ is a \emph{mixed manna}, i.e., relative to an agent, an item may have positive, zero, or negative value. We call an item $o$ \emph{subjective} if $\max_{i\in\N}\val_i(o)\ge0$, and an \emph{objective chore} if $\val_i(o)<0$ for every agent $i\in\N$. These two classes partition $\M$.\footnote{The terminology is asymmetric because in our main result, our algorithm separates only items that are strictly negative for every agent; every subjective item has at least one agent who values it nonnegatively.}

We consider two popular fairness notions in the weighted fair division literature, \emph{weighted envy-freeness} and \emph{weighted maximin share}.

\subsection{Weighted Envy-Freeness}

An allocation is \emph{weighted envy-free} (WEF) if every agent values her own bundle per unit entitlement at least as highly as every other bundle per unit entitlement. In the setting with indivisible goods, \citet{CISZ21} initiated the study of WEF. 
However, as with ordinary envy-freeness in the unweighted setting, exact WEF need not exist: with two agents and one item that both value positively, the agent who does not receive the item has weighted envy. 
This motivated them to propose a weighted extension of a popular relaxation in the unweighted setting, \emph{weighted envy-freeness up to one item} (WEF1), which allows one to delete a good from the recipient's bundle. With indivisible chores, WEF1 analogously allows one to delete a chore from the observer's own bundle~\citep{CISZ21,WZZ25}. The mixed manna definition combines these two options, which was first proposed and studied in \citet{GS24}. We state its definition as follows.

\begin{definition}[WEF1 for mixed manna]\label{def:wef1}
An allocation $A$ is \emph{weighted envy-free up to one item} (WEF1) if, for every ordered pair of agents $i,j \in \N$, at least one of the following holds:
\begin{enumerate}[(i)]
    \item $\frac{\val_i(A_i)}{\wt_i}
        \ge \frac{\val_i(A_j)}{\wt_j}$,
    \item $\frac{\val_i(A_i)}{\wt_i}
        \ge \frac{\val_i(A_j\setminus\{g\})}{\wt_j}$ for some $g\in A_j$ with $\val_i(g)>0$,
    \item $\frac{\val_i(A_i\setminus\{c\})}{\wt_i}
        \ge \frac{\val_i(A_j)}{\wt_j}$
        for some $c\in A_i$ with $\val_i(c)<0$.
\end{enumerate}
\end{definition}

We call $i$ the \emph{observer} and $j$ the \emph{recipient} in the ordered comparison $(i,j)$. Requiring $\val_i(g)>0$ in (ii) and $\val_i(c)<0$ in (iii) does not change the notion. If $\val_i(g)\le0$, then deleting $g$ weakly increases $i$'s value for the recipient's bundle, so (ii) cannot help when (i) fails. Likewise, if $\val_i(c)\ge0$, then deleting $c$ weakly decreases $i$'s value for her own bundle, so (iii) cannot help when (i) fails.

The utilitarian social welfare of $A$ is denoted $\SW(A):=\sum_{i\in\N}\val_i(A_i)$.
For an instance $\mathcal{I}$, let $\OPT(\mathcal{I})$ be the maximum welfare over all allocations, and let $\OPTWEF(\mathcal{I})$ be the maximum welfare over WEF1 allocations. When $\OPTWEF(\mathcal{I})>0$, the ratio $\OPT(\mathcal{I})/\OPTWEF(\mathcal{I})$ measures the welfare loss caused by WEF1.

\subsection{Weighted Maximin Share}

In the unweighted setting, agent $i$'s \emph{maximin share} is the largest value she can ensure by partitioning the items into $n$ bundles and receiving a minimum-value bundle in that partition~\citep{Budish11}. Weighted MMS (WMMS) adapts this property to the setting with unequal entitlements by labeling the bundles by agents and dividing each bundle's value by the entitlement of its label~\citep{FarhadiEtAl19}.

Let $\Pi_n(\M)$ be the set of ordered partitions
$P=(P_1,\dots,P_n)$ of $\M$; empty bundles are allowed.

\begin{definition}[Weighted maximin share]\label{def:wmms}
The \emph{weighted maximin share} of agent $i$ is $\WMMS_i
    := \wt_i
    \max_{P\in\Pi_n(\M)}
    \min_{j\in\N}
    \frac{\val_i(P_j)}{\wt_j}$. 
 An allocation $A$ satisfies \emph{weighted maximin share fairness} (WMMS) if
$\val_i(A_i)\ge\WMMS_i$ for every agent $i$.
\end{definition}
When all entitlements are equal, WMMS reduces to the usual MMS~\citep{Budish11}.

We also consider \emph{fractional Pareto optimality} (fPO), a well-studied efficiency notion in fair division with indivisible items~\citep{AMS20,GS24,Mahara26}. 
It compares an integral allocation with the larger class of allocations in which items may be divided fractionally.

A \emph{fractional allocation} is a matrix $\mathbf{x}=(x_{io})$ with $x_{io}\ge0$ and $\sum_{i\in\N}x_{io}=1$ for every item $o\in\M$. Agent $i$ receives utility $\sum_{o\in\M}x_{io}\val_i(o)$.

\begin{definition}[Fractional Pareto optimality]\label{def:fpo}
An integral allocation $A$ is \emph{fractionally Pareto optimal} (fPO) if no fractional allocation weakly increases every agent's utility and strictly increases at least one agent's utility.
\end{definition}

\section{A Polynomial-Time WEF1 Algorithm for Mixed Manna}\label{sec:wef1}

\citet{GS24} left open whether every additive mixed manna instance with arbitrary positive entitlements admits a complete WEF1 allocation.
Our following result answers this question in the affirmative and gives a polynomial-time algorithm.

\begin{theorem}\label{thm:main}
Every additive mixed manna instance with arbitrary positive entitlements admits a complete WEF1 allocation, computable in polynomial time.
\end{theorem}

The proof is constructive. \Cref{sec:bundling} states the bundling properties taken from earlier work. The main weighted argument begins in \cref{sec:rwps}, where the continuous analysis of RWPS is used to derive the two-sided comparison needed after acceptable bundles are assigned. \Cref{sec:many} gives the maximum-$q$ assignment required by that comparison. \Cref{sec:few} treats the case in which some agents receive no objective chore by combining the few-chore refinement with the weighted goods procedure. \Cref{subsec:wef1-overview} gives the complete algorithm.

\subsection{Combining Items into Acceptable Bundles}\label{sec:bundling}

Rules (B1) and (B2) below are adapted from
\citet[Algorithm~1]{AzizEtAl25Cake}; the corresponding terminal properties are
proved in \citet[Lemma~4.5]{AzizEtAl25Cake}.
The same preprocessing is used by
\citet[Section~3]{AzizEtAl26BoBW}, and
\citet[Section~4.1]{LuMackenzieSuzuki26}
give a related variant with an additional unit upper bound.
We use the terms \emph{acceptable bundle} and \emph{acceptance set} for the
corresponding \emph{meta-good} and \emph{interest set}.

An \emph{acceptable bundle} is a nonempty bundle of original items that at least one agent values nonnegatively. Each acceptable bundle is treated as indivisible until it is \emph{unpacked} at the end.

Initially, every subjective original item is a singleton acceptable bundle, and every objective chore is placed in a set $\Chores$. Starting from this partition, repeatedly apply either operation until neither applies:

\begin{enumerate}[(B1)]
    \item If an agent $i$ values at least two current acceptable bundles nonnegatively, merge all current bundles that $i$ values nonnegatively.
    \item If a current acceptable bundle $B$, an objective chore $c\in\Chores$, and an agent $i$ satisfy $\val_i(B\cup\{c\})\ge0$, merge $c$ into $B$ and remove $c$ from $\Chores$.
\end{enumerate}

Let $\Bundles$ be the final family of acceptable bundles. For $B\in\Bundles$, define its acceptance set $T_B:=\{i\in\N:\val_i(B)\ge 0\}$.
The next lemma states the terminal properties needed in the rest of the proof,
in our notation.

\begin{lemma} \label{lem:bundling}
The bundling process terminates and satisfies the following properties:
\begin{enumerate}[(i)]
    \item $T_B\ne\varnothing$ for every $B\in\Bundles$;
    \item the sets $(T_B)_{B\in\Bundles}$ are pairwise disjoint;
    \item for all $B\in\Bundles$, $c\in\Chores$, and $i\in\N$, combining $B$ with $c$ gives $\val_i(B\cup\{c\})<0$.
\end{enumerate}
Moreover, every current acceptable bundle contains at least one original subjective item.
\end{lemma}

\begin{proof}
Each operation replaces at least two current objects by one, so the process terminates. Both operations preserve the facts that every acceptable bundle has a nonempty acceptance set and contains an original subjective item. If two final acceptance sets intersect, their common agent values two current bundles nonnegatively, so (B1) is still applicable. If (iii) fails, then (B2) is still applicable. Hence all three properties hold at termination.
\end{proof}

\subsection{A Two-Sided Comparison for the Reversed Weighted Picking Sequence}\label{sec:rwps}

The weighted picking sequence for goods was first proposed by \citet{CISZ21}.
A continuous representation of weighted picking sequences for chores is then used by
\citet{LiLiWu22}, and \citet{WZZ25} introduce the reversed weighted picking
sequence (RWPS) and prove its one-sided WEF1 comparison for chores.

Let $\Chores$ contain only objective chores and define the positive cost $d_i(c):=-\val_i(c)>0$. We use RWPS to allocate $\Chores$, but the later assignment of acceptable bundles
requires a different comparison. The usual RWPS comparison deletes the observer's
last actual pick and retains the recipient's entire chore bundle. We instead need
to delete the observer's last actual pick and also set aside the recipient's first
actual pick. \Cref{lem:two-sided} gives exactly this comparison. RWPS has two steps:

\begin{enumerate}
    \item Build a forward schedule of length $|\Chores|$. At each position, select an agent minimizing the number of positions already assigned to her in the forward schedule divided by her entitlement, breaking ties arbitrarily.
    \item Traverse this schedule in reverse. Whenever agent $i$ is selected, she takes a remaining chore of minimum $d_i$-cost.
\end{enumerate}

Let $C_i$ be agent $i$'s RWPS bundle and let $k_i:=|C_i|$. For each agent $i$ with $k_i\ge1$, index the chores as $C_i=\{e_{i,1},\dots,e_{i,k_i}\}$ by increasing position in the \emph{forward} schedule. Consequently, $x_i:=e_{i,1}$ is $i$'s last actual pick, and $y_i:=e_{i,k_i}$ is $i$'s first actual pick.
For such an agent, we refer to $x_i$ as the \emph{last-pick chore} and to $y_i$ as the \emph{first-pick chore}. Define  $q_i:=\frac{k_i-1}{\wt_i}$, the number of chores left in $C_i$ after deleting $x_i$, divided by $i$'s entitlement.

\begin{lemma}\label{lem:two-sided}
Assume $k_i,k_j\ge1$. For every observer $i$ and recipient $j$,
\begin{equation}\label{eq:rwps-standard}
    \frac{d_i(C_i\setminus\{x_i\})}{\wt_i}
    \le
    \frac{d_i(C_j)}{\wt_j}.
\end{equation}
Moreover, if $q_i\le q_j$, then the same comparison remains valid after setting aside $j$'s first-pick chore:
\begin{equation}\label{eq:rwps-reserved}
    \frac{d_i(C_i\setminus\{x_i\})}{\wt_i}
    \le
    \frac{d_i(C_j\setminus\{y_j\})}{\wt_j}.
\end{equation}
\end{lemma}

\begin{proof}
Parameterize the forward schedule by $t\in[0,|\Chores|]$, using one unit interval for each schedule position. During a position assigned to agent $a$, let $s_a(t)$ increase linearly by $1/\wt_a$ and keep every other $s_b(t)$ fixed. Thus, at the beginning of any position, $s_a$ is the number of earlier positions assigned to $a$, divided by $\wt_a$.

For $z\in[k_i]$, the interval of normalized counts from $(z-1)/\wt_i$ to $z/\wt_i$ corresponds to the forward occurrence associated with $e_{i,z}$. On this interval, define the piecewise-constant cost function $\rho_i(\alpha):=d_i(e_{i,z})$  for $\alpha\in
    (\frac{z-1}{\wt_i},\frac z{\wt_i}]$.
From observer $i$'s perspective, define the corresponding function for $j$ by $\rho_j^i(\beta):=d_i(e_{j,z})$ for $\beta\in
    (\frac{z-1}{\wt_j},\frac z{\wt_j}]$.
Both functions are nonnegative.

Consider the time $\tau$ at which $i$ begins her last forward occurrence. At that time $s_i(\tau)=q_i$, and the scheduling rule chooses $i$ with minimum normalized occurrence count. Therefore
\begin{equation}\label{eq:load-bound}
    q_i=s_i(\tau)\le s_j(\tau)\le s_j(|\Chores|)=\frac{k_j}{\wt_j}.
\end{equation}

Fix $\alpha\in(1/\wt_i,k_i/\wt_i)$ that is not an endpoint of one of the intervals above, and suppose $\rho_i(\alpha)=d_i(e_{i,z})$. When this forward occurrence begins, at time $t^*$, $s_i(t^*)=\frac{z-1}{\wt_i}\le s_j(t^*)$.

Moreover, $\alpha-1/\wt_i\le(z-1)/\wt_i$. Hence $j$ reaches normalized count $\alpha-1/\wt_i$ no later than $t^*$. The associated occurrence of $j$ is therefore no later than the occurrence associated with $e_{i,z}$ in the forward schedule, so its chore is picked no earlier in the reverse traversal. Apart from interval endpoints, that chore is still available when $i$ chooses $e_{i,z}$. Since $i$ chooses a least-cost remaining chore,
\begin{equation}\label{eq:pointwise}
    \rho_i(\alpha)
    \le
    \rho_j^i\left(\alpha-\frac1{\wt_i}\right)
    \quad\text{for almost every such }\alpha.
\end{equation}

Integrating~\eqref{eq:pointwise} and changing variables gives
\begin{equation}
    \frac{d_i(C_i\setminus\{x_i\})}{\wt_i}
    =\int_{1/\wt_i}^{k_i/\wt_i}\rho_i(\alpha)\,d\alpha
    \le \int_0^{q_i}\rho_j^i(\beta)\,d\beta.
    \label{eq:rwps-prefix}
\end{equation}
By~\eqref{eq:load-bound}, extending the nonnegative integral to $k_j/\wt_j$ proves~\eqref{eq:rwps-standard}. If $q_i\le q_j$, then instead
\[
    \int_0^{q_i}\rho_j^i(\beta)\,d\beta
    \le
    \int_0^{q_j}\rho_j^i(\beta)\,d\beta
    =\frac{d_i(C_j\setminus\{y_j\})}{\wt_j},
\]
because $q_j=(k_j-1)/\wt_j$ and $y_j=e_{j,k_j}$ is the first actual pick. This proves~\eqref{eq:rwps-reserved}.
\end{proof}

The standard chores guarantee~\eqref{eq:rwps-standard} removes one chore from the observer's own bundle. The additional comparison~\eqref{eq:rwps-reserved} says that, when $q_i\le q_j$, the same inequality still holds after $y_j$ is also omitted from the recipient's chore bundle. We can then place $y_j$ with an acceptable bundle assigned to $j$; by \cref{lem:bundling}(iii), that pair has negative value for observer $i$.

\paragraph{Best possible count condition.}
The condition $q_i\le q_j$ cannot be weakened if \eqref{eq:rwps-reserved} is required to hold for every cost vector. If every chore has cost $1$ to observer $i$, then the left-hand side of \eqref{eq:rwps-reserved} is $q_i$ and the right-hand side is $q_j$. Thus, the comparison holds exactly when $q_i\le q_j$. More generally, \cref{thm:several-rwps} gives the two-sided comparison after several of the observer's last picks are deleted and several of the recipient's first picks are set aside. \Cref{thm:multi-bundle} then gives a WEF1 condition that does not require pairwise disjoint acceptance sets.

\subsection{At Least \texorpdfstring{$n$}{n} Remaining Objective Chores}\label{sec:many}

Suppose that $|\Chores|\ge n$. We first allocate all objective chores by RWPS.
Every agent then has a last actual pick. Afterward, each acceptable bundle is assigned to an agent with maximum $q$-value in its acceptance set.
This choice gives the inequality required by \cref{lem:two-sided} for every observer who values the assigned bundle nonnegatively.

Assume $|\Chores|\ge n$. Allocate the chores in $\Chores$ by RWPS. The first $n$ forward positions contain every agent once, so $k_i\ge1$ for all $i$.

For each $B\in\Bundles$, assign $B$ to an agent $j\in T_B$ whose $q_j$ is maximum within $T_B$; such an agent exists by \cref{lem:bundling}(i). By \cref{lem:bundling}(ii), every agent receives at most one acceptable bundle. Let $B_i$ be the bundle assigned to $i$, or $\varnothing$ if none is assigned, and set $A_i:=C_i\cup B_i$.

\begin{lemma}\label{lem:many}
The allocation described above is WEF1. More strongly, for every agent $i$, deleting the single chore $x_i$ eliminates all of $i$'s weighted envy.
\end{lemma}

\begin{proof}
Fix distinct agents $i$ and $j$. If $B_i$ is nonempty, it is assigned only to an agent in its acceptance set; in either case, $\val_i(B_i)\ge0$.
We consider two cases.

\begin{description}
    \item[Case 1: $B_j=\varnothing$, or $B_j\ne\varnothing$ and $i\notin T_{B_j}$.]
    Then $\val_i(B_j)<0$ in the latter case and $\val_i(B_j)=0$ in the former. Since $d_i(c)=-\val_i(c)$ for every objective chore $c$, multiplying~\eqref{eq:rwps-standard} by $-1$ reverses the inequality and gives us
\[
    \frac{\val_i(C_i\setminus\{x_i\})}{\wt_i}
    \ge
    \frac{\val_i(C_j)}{\wt_j}.
\]
Adding the nonnegative bundle $B_i$ to the left and the nonpositive bundle $B_j$ to the right gives
\[
    \frac{\val_i(A_i\setminus\{x_i\})}{\wt_i}
    \ge
    \frac{\val_i(A_j)}{\wt_j}.
\]

    \item[Case 2: $B_j\ne\varnothing$ and $i\in T_{B_j}$.]
    Since $B_j$ is assigned to an agent in $T_{B_j}$ with largest $q$-value, $q_i\le q_j$.
    The two-sided RWPS lemma therefore lets us set aside $j$'s first-pick chore:
\begin{equation}\label{eq:many-reserve}
    \frac{\val_i(C_i\setminus\{x_i\})}{\wt_i}
    \ge
    \frac{\val_i(C_j\setminus\{y_j\})}{\wt_j}.
\end{equation}
By \cref{lem:bundling}(iii), $\val_i(B_j\cup\{y_j\})<0$.
Consequently,
\begin{equation}\label{eq:many-put-back}
    \val_i(A_j)
    =\val_i(C_j\setminus\{y_j\})
      +\val_i(B_j\cup\{y_j\})
    <\val_i(C_j\setminus\{y_j\}).
\end{equation}
Combining $\val_i(B_i)\ge0$, \eqref{eq:many-reserve}, and~\eqref{eq:many-put-back} proves the desired comparison after deleting $x_i$.
\end{description}
Since $x_i$ is an original objective chore, $\val_i(x_i)<0$, so deleting it is valid under \cref{def:wef1}.
\end{proof}

\subsection{Fewer than \texorpdfstring{$n$}{n} Remaining Objective Chores}\label{sec:few}

Suppose that fewer than $n$ objective chores remain. Some agents will receive no
objective chore, so the RWPS argument from \cref{sec:many} cannot be applied to
every agent. We therefore refine the acceptable bundles before allocating them.

For related good-minimality refinements in the few-chore case, compare
\citet[Algorithm~4 and Lemma~4.11]{AzizEtAl25Cake}
and \citet[Section~5 and Proposition~5.1]{AzizEtAl26BoBW}.
The rules below are stated for the two properties in \cref{lem:refinement} that
are needed by our weighted WEF1 argument.

Assume now that $|\Chores|<n$. There are too few objective chores to give every agent one chore, so we use a different allocation. We refine the acceptable bundles until two statements hold:

\begin{enumerate}
    \item for every nonsingleton acceptable bundle, deleting one original item is at least as effective as deleting the whole bundle in a WEF1 comparison;
    \item every remaining objective chore is more negative than the total value of all acceptable bundles that any one observer values nonnegatively.
\end{enumerate}

The first statement lets us replace deletion of a whole acceptable bundle by deletion of one original item. The second makes the final bundle of every agent who receives a remaining chore negative for every observer.

Let $P$ be the set of original subjective items. Starting from the bundles produced in \cref{sec:bundling}, repeatedly apply the following rules:

\begin{enumerate}[(R1)]
    \item If a nonsingleton acceptable bundle $B$ and an agent $i$ satisfy
    \begin{equation}\label{eq:split-trigger}
        \val_i(B\setminus\{g\})\ge0
        \quad\text{for every }g\in B\cap P,
    \end{equation}
    choose $g\in B\cap P$ with $\val_i(g)\ge0$ and replace $B$ by the two acceptable bundles $\{g\}$ and $B\setminus\{g\}$.

    \item If there are $c\in\Chores$ and $i\in\N$ such that
    \begin{equation}\label{eq:absorb-trigger}
        \val_i(c)+
        \sum_{B\in\Bundles:\,\val_i(B)\ge0}\val_i(B)
        \ge0,
    \end{equation}
    merge $c$ with all current acceptable bundles that $i$ values nonnegatively, remove $c$ from $\Chores$, and return to (R1).
\end{enumerate}

Rule (R1) is well-defined. If every subjective item in $B$ were negative for $i$, then for every $g\in B\cap P$, all items in the nonempty bundle $B\setminus\{g\}$ would be negative for $i$, contradicting~\eqref{eq:split-trigger}. Hence some $g\in B\cap P$ satisfies $\val_i(g)\ge0$. Both resulting pieces are acceptable to $i$. Moreover, $B\setminus\{g\}$ still contains a subjective item; otherwise it would be a nonempty bundle of objective chores and would be strictly negative for $i$.

Rule (R2) is also well-defined. If~\eqref{eq:absorb-trigger} holds, then $\val_i(c)<0$ implies that at least one bundle in the sum has positive value. The merged bundle has nonnegative value for $i$ and contains an original subjective item.

The next lemma states the properties obtained when neither refinement rule applies.

\begin{lemma}%[Properties after refinement]
\label{lem:refinement}
The refinement terminates and gives us the following properties:
\begin{enumerate}[(i)]
    \item for every nonsingleton acceptable bundle $B$ and every observer $i$, there is an original subjective item $g\in B$ such that $\val_i(B\setminus\{g\})<0$;
    \item for every remaining objective chore $c\in\Chores$ and every observer $i$, $-\val_i(c)
        > \sum_{B\in\Bundles:\,\val_i(B)\ge0}\val_i(B)$.
\end{enumerate}
The number of remaining objective chores never increases and remains below $n$.
\end{lemma}

\begin{proof}
Every acceptable bundle contains an original subjective item, so the number of current bundles is at most $|P|$. Rule (R1) leaves $|\Chores|$ unchanged and increases $|\Bundles|$ by one. Rule (R2) decreases $|\Chores|$ by one. Hence the ordered pair $(|\Chores|,\ |P|-|\Bundles|)$ strictly decreases in lexicographic order. The process therefore terminates. At termination, the negation of~\eqref{eq:split-trigger} is exactly (i), and the negation of~\eqref{eq:absorb-trigger} is exactly (ii).
\end{proof}

\subsubsection{A weighted picking sequence for the remaining acceptable bundles}\label{subsec:goods}

A remaining acceptable bundle may be valued negatively by some agents. We need both WEF1 at the level of acceptable bundles and the requirement that every bundle be assigned to an agent who values it nonnegatively. The procedure below obtains these two properties simultaneously. Its proof relates the procedure to the weighted picking sequence by using zero-valued dummy items. 

\begin{lemma}\label{lem:subjective-goods}
Let $H$ be a finite set of indivisible items and let $R\subseteq\N$ be a nonempty set of agents. Suppose every item $h\in H$ has at least one agent $i\in R$ with $\val_i(h)\ge0$. Then $H$ admits an allocation among $R$ with the following properties:
\begin{enumerate}[(i)]
    \item every item is assigned to an agent who values it nonnegatively;
    \item for every $i,j\in R$, either $i$ has no weighted envy toward $j$, or there is an item $h\in A_j$ with $\val_i(h)>0$ such that $\frac{\val_i(A_i)}{\wt_i} \ge \frac{\val_i(A_j\setminus\{h\})}{\wt_j}$.
\end{enumerate}
Such an allocation can be computed in polynomial time.
\end{lemma}

\begin{proof}
If $H=\varnothing$, the empty allocation satisfies the lemma. Otherwise, use the following finite procedure. For each $i\in R$, let $r_i$ be the number of items already assigned to $i$. While an item remains, call an agent \emph{active} if she values at least one remaining item nonnegatively. Select an active agent minimizing $(r_i/\wt_i,i)$ lexicographically, give her a remaining item of maximum $\val_i$-value among those she values nonnegatively, and increase $r_i$ by one. An active agent always exists by the assumption on $H$, so the procedure allocates every item in $|H|$ rounds and satisfies (i).

It remains to prove (ii). Define nonnegative valuations $u_i(h):=\max\{\val_i(h),0\}$. We compare the finite procedure with the usual weighted picking sequence under $u$. Run the forward scheduling rule from \cref{sec:rwps} over all agents in $R$, breaking ties by agent index. Whenever the selected agent is inactive, insert a fresh dummy item worth zero to every agent. Whenever she is active, give her the real item chosen by the finite procedure. In the latter case she has received no earlier dummy: the set of remaining items only shrinks, so an agent who is active now was active at every earlier turn. Her number of earlier schedule positions is therefore $r_i$. After dummy turns are omitted, the selected active agent is exactly the one minimizing $(r_i/\wt_i,i)$ in the finite procedure. Only finitely many dummy turns occur between two real assignments: each such turn increases the selected inactive agent's normalized occurrence count, while the counts of active agents remain fixed.

At each schedule position, the selected agent receives an item of maximum $u_i$-value. If she is active, positive real items keep their value under $u_i$, and the tie-breaking prefers a nonnegative real item to a dummy or to a negatively valued real item. If she is inactive, all remaining real items have $u_i$-value zero, and the dummy is a valid maximum-value choice. Stop when the last real item is assigned, and let $D$ be the finite set of inserted dummies. The resulting allocation of $H\cup D$ is therefore a standard weighted picking-sequence allocation for the nonnegative valuations $u$, so it is WEF1 under $u$ by \citet{CISZ21}. The dummy items are used only in this proof; the finite procedure described above makes exactly $|H|$ assignments and does not create them.

Now remove the dummies. Every real item assigned to $i$ is nonnegative for $i$, so $u_i(A_i)=\val_i(A_i)$. For every real bundle $S$, $u_i(S)\ge\val_i(S)$. Thus any WEF1 inequality under $u$ remains valid under $\val_i$ after the same real item is deleted from the recipient's bundle. If the item deleted under $u$ is a dummy or has $u_i$-value zero, its deletion does not change the right-hand side; the same inequality then gives weighted non-envy under $\val_i$. Hence, whenever a deletion is needed under $\val_i$, it deletes a real item $h$ with $u_i(h)>0$, equivalently $\val_i(h)>0$. This proves (ii).
\end{proof}

\subsubsection{Allocating the remaining chores and acceptable bundles}\label{subsec:few-allocation}

The ordered step in which selected agents take all still unassigned acceptable bundles that they value nonnegatively also appears in
\citet[Section~5]{AzizEtAl26BoBW}.
Here this step is followed by the weighted picking sequence from
\cref{lem:subjective-goods}. In the proof of \cref{lem:few}, \cref{lem:refinement}(i) converts deletion of an acceptable bundle into deletion of one original item, while \cref{lem:refinement}(ii) handles comparisons involving
agents who receive objective chores.

Let $t:=|\Chores|<n$. Choose any set $T\subseteq\N$ of size $t$, assign the $t$ objective chores bijectively to the agents in $T$, and write $c_i$ for the chore assigned to $i\in T$. In any order, let each agent $i\in T$ take every still-unassigned acceptable bundle that she values nonnegatively. Let $R:=\N\setminus T$.

Every acceptable bundle that remains is valued nonnegatively by some agent in $R$: otherwise, because the bundle is acceptable, some agent in $T$ would value it nonnegatively and would have taken it. We may therefore allocate the remaining bundles among $R$ using \cref{lem:subjective-goods}.

\begin{lemma}\label{lem:few}
After unbundling the acceptable bundles, the resulting allocation is WEF1.
\end{lemma}

\begin{proof}
We consider separately observers who receive an objective chore and observers who do not.

\paragraph{Case 1: the observer receives one objective chore.}
Fix $i\in T$. Since $c_i$ is an objective chore, $\val_i(c_i)<0$. Every acceptable bundle assigned to $i$ is nonnegative for her, so $\val_i(A_i\setminus\{c_i\})\ge0$.
Every acceptable bundle left for $R$ is negative for $i$, because $i$ would otherwise have taken it in the preceding assignment to the agents in $T$. Moreover, for any observer $h$ and any chore holder $j\in T$, \cref{lem:refinement}(ii) gives
\begin{equation}\label{eq:choreholder-negative}
    \val_h(A_j)
    \le
    \val_h(c_j)+
    \sum_{B\in\Bundles:\,\val_h(B)\ge0}\val_h(B)
    <0.
\end{equation}
Thus, after deleting $c_i$, agent $i$ has value at least zero, while every other chore holder and every agent in $R$ has a bundle of value at most zero from $i$'s perspective. Hence removing $c_i$ establishes the WEF1 inequality against every recipient.

\paragraph{Case 2: the observer receives no objective chore.}
Fix $i\in R$. She receives only acceptable bundles that she values nonnegatively, so $\val_i(A_i)\ge0$. By~\eqref{eq:choreholder-negative}, every chore holder's bundle has negative value to $i$; therefore $i$ does not envy any agent in $T$. Among agents in $R$, \cref{lem:subjective-goods} gives WEF1 at the level of acceptable bundles.

It remains to turn a deleted acceptable bundle into one original item. Suppose the bundle-level WEF1 comparison against recipient $j\in R$ deletes $B\subseteq A_j$. By \cref{lem:subjective-goods}, $\val_i(B)>0$. If $B$ is a singleton, it is therefore a positively valued original item. Otherwise, \cref{lem:refinement}(i) gives an original subjective item $g\in B$ such that $\val_i(B\setminus\{g\})<0$. Hence $\val_i(g)>0$, and
\[
    \val_i(A_j\setminus\{g\})
    =\val_i(A_j\setminus B)+\val_i(B\setminus\{g\})
    <\val_i(A_j\setminus B).
\]
Deleting $g$ is therefore at least as effective as deleting the whole acceptable bundle $B$, so the bundle-level WEF1 inequality implies a WEF1 inequality after deleting one original item.
\end{proof}

\subsection{Complete Algorithm and Proof}\label{subsec:wef1-overview}

\Cref{alg:wef1} combines the two allocation cases proved above. All choices not fixed in the description may be made arbitrarily; the preceding lemmas use only the stated conditions.

\begin{algorithm}[ht]
\caption{Complete WEF1 allocation for mixed manna}\label{alg:wef1}
\small
\begin{algorithmic}[1]
\Require An instance $(\N,\M,(\wt_i)_{i\in\N},(\val_i)_{i\in\N})$ as in \cref{sec:model}.
\State $P\gets\{o\in\M:\max_{i\in\N}\val_i(o)\ge0\}$, $\Bundles\gets\{\{o\}:o\in P\}$, and $\Chores\gets\M\setminus P$.
\While{rule (B1) or (B2) from \cref{sec:bundling} applies}
    \If{some $i$ values at least two bundles in $\Bundles$ nonnegatively}
        \State Choose such an $i$ and replace all bundles $B$ with $\val_i(B)\ge0$ by their union.
    \Else
        \State Choose $B\in\Bundles$, $c\in\Chores$, and $i\in\N$ with $\val_i(B\cup\{c\})\ge0$.
        \State Replace $B$ by $B\cup\{c\}$ and remove $c$ from $\Chores$.
    \EndIf
\EndWhile
\If{$|\Chores|\ge n$}
    \State Build a forward schedule of length $|\Chores|$. At each position, choose an agent minimizing her number of earlier positions divided by her entitlement, breaking ties by agent index.
    \State Traverse the schedule in reverse. Whenever agent $i$ occurs, give her a remaining chore maximizing $\val_i(c)$, equivalently minimizing $d_i(c)$. Let $C_i$ be her chore bundle and $k_i:=|C_i|$.
    \State Set $q_i\gets(k_i-1)/\wt_i$ and $A_i\gets C_i$ for every $i\in\N$.
    \For{each $B\in\Bundles$}
        \State Choose $j\in\argmax\{q_h:h\in\N,\ \val_h(B)\ge0\}$ and set $A_j\gets A_j\cup B$.
    \EndFor
\Else
    \While{rule (R1) or (R2) from \cref{sec:few} applies}
        \If{some nonsingleton $B\in\Bundles$ and $i\in\N$ satisfy $\val_i(B\setminus\{g\})\ge0$ for every $g\in B\cap P$}
            \State Choose such $B,i$ and some $g\in B\cap P$ with $\val_i(g)\ge0$; replace $B$ by $\{g\}$ and $B\setminus\{g\}$.
        \Else
            \State Choose $c\in\Chores$ and $i\in\N$ with $\val_i(c)+\sum_{B\in\Bundles:\,\val_i(B)\ge0}\val_i(B)\ge0$.
            \State Merge $c$ with all bundles $B$ satisfying $\val_i(B)\ge0$, and remove $c$ from $\Chores$.
        \EndIf
    \EndWhile
    \State Set $t\gets|\Chores|$; choose $T\subseteq\N$ with $|T|=t$ and assign the chores bijectively as $(c_i)_{i\in T}$.
    \State Set $R\gets\N\setminus T$, $A_i\gets\{c_i\}$ for $i\in T$, and $A_i\gets\varnothing$ for $i\in R$.
    \For{each $i\in T$, in any order}
        \State Give $i$ every still-unassigned $B\in\Bundles$ with $\val_i(B)\ge0$.
    \EndFor
    \State Set $r_i\gets0$ for every $i\in R$.
    \While{some acceptable bundle remains unassigned}
        \State Choose an agent $i\in R$ who values a remaining bundle nonnegatively and minimizes $(r_i/\wt_i,i)$ lexicographically.
        \State Give $i$ a remaining bundle of maximum $\val_i$-value among those she values nonnegatively, and increment $r_i$.
    \EndWhile
\EndIf
\State Unpack every acceptable bundle and \Return $A$.
\end{algorithmic}
\end{algorithm}

\begin{proof}[Proof of \cref{thm:main}]
Run \cref{alg:wef1}. If $|\Chores|\ge n$, \cref{lem:many} proves that its output is WEF1. If $|\Chores|<n$, \cref{lem:few} proves the same. The two branches are exhaustive, and both allocate every original item.

It remains to bound the running time. The initial bundling step performs at most $m$ merges. In the few-chore refinement, (R2) is applied at most $m$ times. For each fixed value of $|\Chores|$, rule (R1) is applied at most $|P|$ times because it increases $|\Bundles|$ and $|\Bundles|\le|P|$. Thus the refinement performs $O(m^2)$ operations. RWPS uses $|\Chores|$ rounds, and the direct goods implementation in \cref{lem:subjective-goods} uses one round per remaining acceptable bundle. All operations are additions, comparisons, and selections over polynomially many rational quantities. Every bundle value is a sum of at most $m$ input values, and every normalized count uses an integer at most $m$ and one input entitlement, so these quantities have polynomial binary encoding length. Thus, the complete algorithm runs in polynomial time.
\end{proof}

\subsection{The Utilitarian Price of WEF1 is Unbounded}\label{subsec:price}

With equal entitlements, WEF1 coincides with EF1, so the known EF1 prices give the relevant comparison. For indivisible goods with normalized valuations, the utilitarian price of EF1 is $\Theta(\sqrt{n})$~\citep{BeiEtAl19,BarmanBhaskarShah20}. For indivisible chores, it is $5/4$ for two agents and unbounded from three agents onward~\citep{WZZ25}. Mixed manna already has an unbounded price with two agents, even when the agents agree on the sign of every item and both value the whole set at one.

The existence theorem therefore does not imply any fixed welfare approximation. We prove this under a fixed four-value domain.

\begin{theorem}\label{thm:price}
For every integer $k\ge2$, there is a two-agent additive mixed manna instance $\mathcal{I}_k$ with equal entitlements such that (i) $\val_1(\M)=\val_2(\M)=1$; (ii) every item is a good for both agents or a chore for both agents; (iii) every singleton value belongs to
    $\{-2/3,-1/3,1/3,2/3\}$; (iv) $\OPT(\mathcal{I}_k)=\frac{k+4}{3}$ and $\OPTWEF(\mathcal{I}_k)=\frac53$.

Moreover, a WEF1 and fPO allocation attains welfare $5/3$. Hence $\frac{\OPT(\mathcal{I}_k)}{\OPTWEF(\mathcal{I}_k)}
    =\frac{k+4}{5}\longrightarrow\infty$.
The conclusion is unchanged if the fair allocation is also required to be fPO.
\end{theorem}

\begin{proof}
There are $k+1$ ordinary goods $g_1,\dots,g_{k+1}$, $k$ ordinary chores
$c_1,\dots,c_k$, and one balancing good $b$. Their values are
\[
\begin{array}{c|ccc}
 & g_\ell & c_\ell & b\\ \hline
\val_1 & \frac23 & -\frac23 & \frac13\\[1mm]
\val_2 & \frac13 & -\frac13 & \frac23
\end{array}
\]
for every applicable index $\ell$. Both total values sum to one:
\[
    \val_1(\M)
    =(k+1)\frac23-k\frac23+\frac13=1, \quad \text{and} \quad \val_2(\M)
    =(k+1)\frac13-k\frac13+\frac23=1.
\]

Without a fairness constraint, assign every ordinary good to agent $1$, every ordinary chore to agent $2$, and $b$ to agent $2$. This maximizes each item's contribution separately, so
\[
    \OPT(\mathcal{I}_k)
    =(k+1)\frac23-k\frac13+\frac23
    =\frac{k+4}{3}.
\]

Consider any allocation. Let $x$ be the number of ordinary goods assigned to agent $1$, let $y$ be the number of ordinary chores assigned to agent $1$, and let $d:=x-y$.
Let $s=1$ if $b$ is assigned to agent $1$, and let $s=0$ otherwise. Agent $2$ values agent $1$'s bundle at $\val_2(A_1)=\frac d3+\frac{2s}{3}$.
Since $\val_2(\M)=1$, she values her own bundle at $\val_2(A_2)=1-\frac d3-\frac{2s}{3}$.
Her envy gap is therefore
\begin{equation}\label{eq:price-envy-gap}
    \val_2(A_1)-\val_2(A_2)
    =
    \frac{2d}{3}+\frac{4s}{3}-1.
\end{equation}

Suppose first that $s=0$. Deleting one item can improve agent $2$'s comparison by at most $1/3$: one may delete an ordinary good from $A_1$ or an ordinary chore from $A_2$. If $d\ge3$, the gap in~\eqref{eq:price-envy-gap} is at least $1$, so one deletion cannot remove the envy. Thus every WEF1 allocation with $s=0$ satisfies $d\le2$.

Now suppose that $s=1$. The largest possible improvement is $2/3$, obtained by deleting $b$ from $A_1$. If $d\ge1$, the gap in~\eqref{eq:price-envy-gap} is at least $1$. Hence every WEF1 allocation with $s=1$ satisfies $d\le0$.

The welfare of the allocation is
\[
    \SW(A)
    = \left(\frac{2d}{3}+\frac{s}{3}\right)+\left(1-\frac d3-\frac{2s}{3}\right)
    =1+\frac d3-\frac s3.
\]
Since $d \leq 2$ when $s = 0$ and $d \leq 0$ when $s = 1$ (as shown earlier),
\[
    \SW(A)\le
    \begin{cases}
        5/3,&s=0,\\
        2/3,&s=1.
    \end{cases}
\]
Thus every WEF1 allocation has welfare at most $5/3$.

To attain this value, give agent $1$ all $k+1$ ordinary goods and $k-1$ ordinary chores. Give agent $2$ the remaining ordinary chore and $b$. Agent $1$ has utility $4/3$ and values agent $2$'s bundle at $-1/3$, so she has no envy. Agent $2$ has utility $1/3$ and values agent $1$'s bundle at $2/3$. Deleting any ordinary good from agent $1$'s bundle lowers that value to $1/3$. The allocation is WEF1 and has welfare $5/3$.

Finally, this allocation maximizes the positive weighted objective $\val_1(A_1)+2\val_2(A_2)$ even over fractional allocations. For every ordinary good, assignment to either agent contributes $2/3$ to this objective. The same tie holds for every ordinary chore, while $b$ contributes strictly more when assigned to agent $2$. A fractional Pareto improvement would strictly increase this positive weighted objective, which is impossible. The allocation is therefore fPO.
\end{proof}

The welfare loss requires many ordinary items to be assigned to the agent who values them less. When $b$ is assigned to agent $2$, their number is $(k+1-x)+y=k+1-d\ge k-1$.
Moreover, every WEF1 allocation attaining welfare $5/3$ has $s=0$ and $d=2$, so this number is exactly $k-1$.

\section{WMMS for Equal-Magnitude Mixed Values}\label{sec:wmms}

We now turn to WMMS, the share-based fairness notion for unequal entitlements introduced by \citet{FarhadiEtAl19}. Unlike WEF1, WMMS compares each agent's bundle with a partition benchmark defined from her own valuation. General additive mixed manna admits no fixed positive multiplicative MMS guarantee even with equal entitlements~\citep{KulkarniMehtaTaki20}, so exact WMMS cannot hold without additional structure. 
We identify a valuation class under mixed manna for which exact WMMS is guaranteed for arbitrary entitlements and any number of items. To the best of our knowledge, this is the first result with both features. 
For equal entitlements and normalized values in $\{-1,0,1\}$,
\citet{Hsu24} proves the corresponding MMS existence result.
Our ceiling formula specializes to the unweighted integer target when all entitlements are equal. With arbitrary entitlements, however, the bundles in an agent's WMMS partition have different entitlement-dependent requirements.
The argument below derives these requirements exactly and proves that the targets of all agents can be met in one allocation.

The technical issue is simultaneous allocation. The formula for one agent's WMMS gives an agent-specific target, but agents may disagree on the signs of the items, so partitions attaining their individual WMMS values do not directly give one allocation that meets every target. We show that the ceiling formula implies Hall's inequalities for the positive targets and, after the positively valued items are assigned, also implies that the common chores can be distributed without taking any agent below her target.

\subsection{Existence and Polynomial-Time Computability}\label{subsec:equal-magnitude}

In the unweighted setting, normalized ternary values in $\{-1,0,1\}$ have been studied under envy-based fairness notions for mixed manna~\citep{AleksandrovWalsh20}. The valuation class we consider is precisely this domain after allowing a separate positive scaling for each agent. We call it \emph{equal-magnitude mixed values} because all nonzero singleton values of a given agent have the same absolute value, although that value may differ across agents.

\begin{definition}[Equal-magnitude mixed values]\label{def:equal-magnitude}
An instance has \emph{equal-magnitude mixed values} if, for every agent $i$, there is a number $a_i>0$ such that $\val_i(o)\in\{-a_i,0,a_i\}$ for every item $o$.
\end{definition}

If agent $i$ has a nonzero singleton value, $a_i$ is the common absolute value of her nonzero singleton values; if all her singleton values are zero, we set $a_i:=1$.
Binary goods and binary chores are special cases,\footnote{The setting with binary goods has been extensively studied in the (weighted) fair division literature, see e.g., \cite{BrandlSuTe26,HalpernPrPs20,SuksompongTe22,neoh2025efx}.} and exact WMMS is known in both settings~\citep{ChenEtAl26,AzizChanLi19}. In the mixed manna setting, positive and negative values may occur in the same instance.

For each agent $i$, define the normalized valuation $u_i(o):=\frac{\val_i(o)}{a_i}\in\{-1,0,1\}$ and let $R_i:=u_i(\M)=\sum_{o\in\M}u_i(o)$.
Scaling one agent's valuation by a positive constant scales her WMMS by the same constant, so it suffices to understand the normalized values.

\begin{theorem} \label{thm:wmms-main}
Every equal-magnitude mixed instance with arbitrary positive entitlements admits an exact WMMS allocation, computable in polynomial time.
Moreover, the allocation maximizes $\sum_{i\in\N}u_i(A_i) = \sum_{i\in\N}\frac{\val_i(A_i)}{a_i}$
over all fractional allocations. It is therefore fPO. If $a_1=\cdots=a_n$, it also maximizes ordinary utilitarian welfare.
\end{theorem}

We prove \cref{thm:wmms-main} in the next two subsubsections. The first derives the WMMS formula; the second uses the resulting integer targets to construct the allocation.

\subsubsection{A Formula for the WMMS Value}\label{subsec:wmms-formula}

For agent $i$, the formula depends only on her normalized total value $R_i$. Since normalized bundle values are integers, it also gives the integer utility target used by the allocation algorithm.
For an integer $R$, define
\begin{equation}\label{eq:lambda-def}
    \lambda(R)
    :=
    \max\left\{
        \lambda\in\R:
        \sum_{j\in\N}\left\lceil \lambda\wt_j\right\rceil
        \le R
    \right\}.
\end{equation}
The entitlement vector $(\wt_j)_{j\in\N}$ is fixed, so we write $\lambda(R)$ without displaying it as a second argument. The inequality in \eqref{eq:lambda-def} is necessary if every bundle labeled $j$ is to have an integer value of at least $\lambda\wt_j$.

The maximum in \eqref{eq:lambda-def} is well-defined. Sufficiently negative values of $\lambda$ are feasible, while every feasible value satisfies
\[
    \lambda=\sum_{j\in\N}\lambda\wt_j
    \le \sum_{j\in\N}\lceil\lambda\wt_j\rceil
    \le R.
\]
As $\lambda$ increases, the ceiling sum changes only immediately to the right of a value at which some $\lambda\wt_j$ is an integer; at that value it still equals its value from the left. Hence the supremum of the feasible set is feasible, so the maximum is attained.
For $R=R_i$, this necessary condition is also sufficient for the normalized valuation $u_i$, as the next lemma shows.

\begin{lemma}\label{lem:wmms-formula}
For every agent $i$, $\WMMS_i = a_i\wt_i\lambda(R_i)$.
Consequently, since $u_i(A_i)$ is an integer, agent $i$ receives at least her WMMS exactly when $u_i(A_i)\ge \left\lceil \wt_i\lambda(R_i)\right\rceil$.
\end{lemma}

\begin{proof}
Fix agent $i$ and an ordered partition
$P=(P_1,\dots,P_n)$. Write $z_j:=u_i(P_j)\in\mathbb{Z}$.
The values of the bundles sum to $\sum_{j\in\N}z_j=R_i$.
Let $\lambda_P:=\min_{j\in\N}\frac{z_j}{\wt_j}$.
For every $j$, integrality gives $z_j\ge\left\lceil\lambda_P\wt_j\right\rceil$.
Thus $\sum_{j\in\N}\left\lceil\lambda_P\wt_j\right\rceil \le R_i$, so $\lambda_P\le\lambda(R_i)$. No partition can obtain a larger normalized minimum.

For the converse, set $q_j:=\left\lceil\lambda(R_i)\wt_j\right\rceil$.
By~\eqref{eq:lambda-def}, $\sum_{j\in\N}q_j\le R_i$.

Suppose first that $R_i\ge0$. Then $\lambda(R_i)\ge0$, because $\lambda=0$ is feasible in~\eqref{eq:lambda-def}; hence every $q_j\ge0$. Increase some of the $q_j$ by integer amounts until obtaining nonnegative integers
$z_1,\dots,z_n$ with $z_j\ge q_j$ and $\sum_j z_j=R_i$.
Pair each item of value $-1$ with a distinct item of value $+1$. Exactly $R_i$ positive items remain unpaired. Put $z_j$ of these items in bundle $P_j$, and distribute the zero-valued pairs and the zero-valued items arbitrarily. This realizes the desired bundle values.

If $R_i<0$, then $\lambda(R_i)<0$, so every $q_j\le0$. Starting from the $q_j$, increase entries without making any entry positive until their sum is $R_i$. This is possible because $\sum_{j\in\N}q_j\le R_i$ and the total possible increase up to zero is at least $R_i-\sum_jq_j$. Pair every positive item with a distinct negative item. Exactly $-R_i$ negative items remain unpaired. Put $-z_j$ of these items in bundle $P_j$, and distribute the zero-valued pairs and zero-valued items arbitrarily. This realizes the desired bundle values.

In both cases, the constructed partition satisfies $\frac{z_j}{\wt_j} \ge \frac{q_j}{\wt_j} \ge \lambda(R_i)$ for every $j$.
Thus $\lambda(R_i)$ is attainable. Multiplying the normalized minimum by $a_i\wt_i$ proves our result.
\end{proof}

For later use, define $T_j(R):=\left\lceil \wt_j\lambda(R)\right\rceil$.
Then (i) $T_j(R)$ is nondecreasing in $R$; (ii) $T_j(R)\ge0$ when $R\ge0$; and (iii) $\sum_{j\in\N}T_j(R) \le R$.
The first follows because the feasible set in~\eqref{eq:lambda-def} grows with $R$; the second follows because $\lambda=0$ is feasible for $R\ge0$; and the third is the defining inequality.

The value $\lambda(R)$ can be computed in polynomial time. At the largest feasible value, at least one product $\lambda(R)\wt_j$ is an integer; otherwise $\lambda(R)$ could be increased slightly without changing any ceiling. Let $q$ be the largest integer among these integral products. If $R\ge0$, then $0\le\lambda(R)\le R$, so $q\in[0,R]\subseteq[0,m]$. Suppose $R<0$ and $q<R$. Increasing $\lambda(R)$ slightly raises every integral ceiling by one and leaves every other ceiling unchanged. Since $\lambda(R)<0$, all the resulting ceilings are nonpositive, and every raised ceiling is at most $q+1\le R$. Their sum is therefore at most $q+1\le R$, contradicting maximality. Thus $q\in[R,0]\subseteq[-m,0]$. It follows that $\lambda(R)=q/\wt_j$ for some $j\in\N$ and $q\in[-m,m]$. We can enumerate these $O(nm)$ candidates, evaluate the ceiling sum in \eqref{eq:lambda-def} for each one, and select the largest feasible candidate.

\subsubsection{A Flow Algorithm}\label{subsec:wmms-flow}

We now use the formula to construct the allocation in \cref{thm:wmms-main}. For each agent $i$, let $t_i:=\lceil\wt_i\lambda(R_i)\rceil$. By \cref{lem:wmms-formula}, it is enough to ensure $u_i(A_i)\ge t_i$. A bipartite matching assigns items valued $+1$ so that every positive target is met. The remaining steps distribute the common chores without lowering any agent below her target and assign every item to an agent with maximum normalized value.

Partition the items according to their largest normalized value:
\begin{align*}
    G&:=\{o\in\M:\max_{i\in\N}u_i(o)=1\},\\
    Z&:=\{o\in\M:\max_{i\in\N}u_i(o)=0\},\\
    C&:=\{o\in\M:u_i(o)=-1\text{ for every }i\in\N\}.
\end{align*}
Thus every item in $G$ is valued $+1$ by at least one agent; every item in $Z$ is valued $0$ by at least one agent and is not valued $+1$ by any agent; and every item in $C$ is a common chore.

For $S\subseteq\N$, let $G(S):=
    \{o\in G:\text{some }i\in S\text{ has }u_i(o)=1\}$.
    Then, the algorithm is defined as follows (\Cref{alg:wmms}).

\begin{algorithm}[h]
\caption{Exact WMMS for equal-magnitude mixed values}\label{alg:wmms}
\begin{algorithmic}[1]
\Require An equal-magnitude mixed instance $(\N,\M,(\wt_i)_{i\in\N},(\val_i)_{i\in\N})$ as in \cref{def:equal-magnitude}.
\State Normalize each valuation to $u_i=\val_i/a_i$.
\State Compute $R_i$, $\lambda(R_i)$, and
$t_i=\lceil\wt_i\lambda(R_i)\rceil$ for every agent.
\State For every agent $i$, create $\max\{t_i,0\}$ copies, each adjacent to the items she values $+1$, and find a matching that assigns one distinct item to every copy.
\State Assign every remaining item in $G$ to any agent who values it $+1$.
\State If agent $i$ has received $g_i$ items from $G$, assign all items in $C$ so that she receives at most $g_i-t_i$ of them.
\State Assign every item in $Z$ to an agent who values it zero.
\end{algorithmic}
\end{algorithm}

\begin{proof}[Proof of \cref{thm:wmms-main}]
We first show that the matching in \cref{alg:wmms} exists. By Hall's theorem~\citep{Hall35}, applied after replacing agent $i$ by $\max\{t_i,0\}$ identical copies, it is enough to prove that
\begin{equation}\label{eq:hall-demand}
    \sum_{i\in S}\max\{t_i,0\}\le |G(S)|
    \quad\text{for every }S\subseteq\N.
\end{equation}
Fix $S\subseteq\N$. If $i\in S$ and $t_i>0$, then every item that $i$ values $+1$ belongs to $G(S)$, while negatively valued items can only decrease $R_i$. Therefore $R_i\le |G(S)|$. Since $T_j(R)$ is nondecreasing in $R$, $t_i=T_i(R_i)\le T_i(|G(S)|)$.
Since $|G(S)|\ge0$, properties (ii) and (iii) of $T_j$ give
\[
    \sum_{i\in S}\max\{t_i,0\}
    =\sum_{i\in S:\,t_i>0}t_i
    \le\sum_{i\in S:\,t_i>0}T_i(|G(S)|)
    \le\sum_{j\in\N}T_j(|G(S)|)
    \le |G(S)|.
\]
This proves~\eqref{eq:hall-demand}. Hence the matching assigns every agent at least $\max\{t_i,0\}$ liked items. Assign every remaining item in $G$ to an agent who likes it, and let $g_i$ be the number of items from $G$ assigned to agent $i$. Then
\begin{equation}\label{eq:chore-capacity-individual}
    g_i\ge\max\{t_i,0\}\ge t_i
    \quad\text{for every }i\in\N.
\end{equation}

We next show that the quantities in~\eqref{eq:chore-capacity-individual} are sufficient for all common chores. Agent $i$ likes at most $|G|$ items and dislikes every item in $C$; she may dislike other items as well. Thus, $R_i\le |G|-|C|$.
Using again the monotonicity of $T_i$ and the inequality $\sum_{j\in\N}T_j(R)\le R$, we obtain
\[
    \sum_{i\in\N}t_i
    =\sum_{i\in\N}T_i(R_i)
    \le\sum_{i\in\N}T_i(|G|-|C|)
    \le |G|-|C|.
\]
All items in $G$ have been assigned, so $\sum_i g_i=|G|$. Therefore
\[
    \sum_{i\in\N}(g_i-t_i)
    =|G|-\sum_{i\in\N}t_i
    \ge |C|.
\]
The nonnegative integer bounds $g_i-t_i$ thus allow all items in $C$ to be assigned so that agent $i$ receives at most $g_i-t_i$ of them. Let $c_i$ be the number assigned to $i$. Assign each item in $Z$ to an agent who values it zero; these items do not change the recipients' normalized utilities. 
Consequently, $u_i(A_i)=g_i-c_i\ge t_i$.
By \cref{lem:wmms-formula}, every agent receives at least her WMMS.

It remains to prove efficiency. Every item is assigned to an agent with maximum normalized value: items in $G$ contribute $1$, items in $Z$ contribute $0$, and items in $C$ contribute $-1$. Hence
\begin{equation}\label{eq:max-normalized-welfare}
    \sum_{i\in\N}u_i(A_i)
    =\sum_{o\in\M}\max_{i\in\N}u_i(o)
    =|G|-|C|.
\end{equation}
In a fractional allocation, each item's contribution to total normalized welfare is a convex combination of its normalized values and is therefore at most its maximum normalized value. Thus~\eqref{eq:max-normalized-welfare} is the largest possible fractional value of $\sum_{i\in\N}u_i(A_i)
    =\sum_{i\in\N}\frac{\val_i(A_i)}{a_i}$.
The coefficients $1/a_i$ are strictly positive, so any fractional Pareto improvement would strictly increase this sum, a contradiction. The allocation is therefore fPO. If all $a_i$ are equal, the sum is a positive multiple of ordinary utilitarian welfare.

For each distinct $R_i$, computing $\lambda(R_i)$ requires evaluating $O(nm)$ candidates. The only other nontrivial step is a bipartite flow of polynomial size; all remaining assignments are direct. Thus the algorithm runs in polynomial time.
\end{proof}

The proof treats positive and negative items in two steps. The formula for the targets gives Hall's inequalities for assigning the required $+1$ items. After those assignments, the same formula shows that the total of the bounds $g_i-t_i$ is at least $|C|$, so all common chores can be assigned without reducing any agent below her target.

\subsection{WEF1 and WMMS for Equal-Magnitude Mixed Values}\label{subsec:wef1-wmms-mixed}

Under unequal entitlements, WEF1 does not generally guarantee WMMS: even for goods, it gives no positive multiplicative approximation of WMMS~\citep[Proposition~6.2]{CISZ21}. Structured valuation classes behave differently. For binary goods, \citet[Lemma~61]{GS26Relations} show that WEF1 implies WMMS for arbitrary entitlements; under equal entitlements, their implication results also show that EF1 entails MMS for additive values in $\{-1,0,1\}$. The next theorem determines the guarantee when positive and negative values coexist in the equal-magnitude class.
Let $\wt_{\max}:=\max_{j\in\N}\wt_j$.

\begin{theorem}\label{thm:wef1-wmms-mixed}
Let $A$ be a WEF1 allocation of an equal-magnitude mixed instance. Then every agent $i$ satisfies
\begin{equation}\label{eq:wef1-wmms-integer}
    u_i(A_i)
    \ge
    \left\lceil
        \wt_i\lambda(R_i)-1+\frac{\wt_i}{\wt_{\max}}
    \right\rceil.
\end{equation}
Consequently,
\begin{equation}\label{eq:wef1-wmms-additive}
    \val_i(A_i)
    \ge
    \WMMS_i-a_i\left(1-\frac{\wt_i}{\wt_{\max}}\right).
\end{equation}
For every entitlement vector and every agent $i$ with $\wt_i<\wt_{\max}$, equality holds in both bounds for some instance with identical valuations over one common good and two common chores.
\end{theorem}

\begin{proof}
Fix an agent $i$. By the definition of $\lambda(R_i)$,
\begin{equation}\label{eq:wef1-wmms-rounding}
    \sum_{j\in\N}\left\lceil\lambda(R_i)\wt_j\right\rceil
    \le R_i
    =\sum_{j\in\N}u_i(A_j).
\end{equation}

For every $j\ne i$, the three alternatives in the definition of WEF1 imply
\begin{equation}\label{eq:wef1-pair-bound}
    u_i(A_j)
    \le
    \left\lfloor
        \frac{\wt_j}{\wt_i}u_i(A_i)
        +\max\left\{1,\frac{\wt_j}{\wt_i}\right\}
    \right\rfloor.
\end{equation}
Indeed, weighted non-envy gives
$u_i(A_j)\le(\wt_j/\wt_i)u_i(A_i)$. If WEF1 removes a positively valued item from $A_j$, then that item's normalized value is $1$, giving
$u_i(A_j)\le(\wt_j/\wt_i)u_i(A_i)+1$. If WEF1 removes a negatively valued item from $A_i$, then its normalized value is $-1$, giving
$u_i(A_j)\le(\wt_j/\wt_i)(u_i(A_i)+1)$. Since $u_i(A_j)$ is an integer, these three bounds give~\eqref{eq:wef1-pair-bound}.

Suppose, for a contradiction, that
\begin{equation}\label{eq:wef1-wmms-contradiction}
    u_i(A_i)
    <
    \wt_i\lambda(R_i)-1+\frac{\wt_i}{\wt_{\max}}.
\end{equation}
First consider an agent $j$ with $\wt_j\ge\wt_i$. By~\eqref{eq:wef1-pair-bound}, $u_i(A_j)
    \le
    \lfloor
        \frac{\wt_j}{\wt_i}(u_i(A_i)+1)
    \rfloor$,
and~\eqref{eq:wef1-wmms-contradiction} gives us
\[
    \frac{\wt_j}{\wt_i}(u_i(A_i)+1)
    <
    \lambda(R_i)\wt_j+\frac{\wt_j}{\wt_{\max}}
    \le
    \lambda(R_i)\wt_j+1.
\]
Since $u_i(A_j)$ is an integer, we have that $u_i(A_j)\le\left\lceil\lambda(R_i)\wt_j\right\rceil$.

Now consider $j$ with $\wt_j<\wt_i$. Again by~\eqref{eq:wef1-pair-bound}, $u_i(A_j)
    \le
    \lfloor
        \frac{\wt_j}{\wt_i}u_i(A_i)+1
    \rfloor$.
Moreover,
\begin{equation*}
    \frac{\wt_j}{\wt_i}u_i(A_i)+1 
    < \lambda(R_i)\wt_j
    +1-\frac{\wt_j}{\wt_i}
        \left(1-\frac{\wt_i}{\wt_{\max}}\right)\le
    \lambda(R_i)\wt_j+1.
\end{equation*}
Thus, $u_i(A_j)\le\left\lceil\lambda(R_i)\wt_j\right\rceil$ in this case as well.

Finally, $\wt_i/\wt_{\max}\le1$, so~\eqref{eq:wef1-wmms-contradiction} implies
$u_i(A_i)<\lambda(R_i)\wt_i$. Since $u_i(A_i)$ is an integer,
\[
    u_i(A_i)
    \le
    \left\lceil\lambda(R_i)\wt_i\right\rceil-1.
\]
Summing the bounds for all bundles and using~\eqref{eq:wef1-wmms-rounding} gives
\[
    R_i
    =\sum_{j\in\N}u_i(A_j)
    \le
    \sum_{j\in\N}\left\lceil\lambda(R_i)\wt_j\right\rceil-1
    \le R_i-1,
\]
a contradiction. Therefore
$u_i(A_i)\ge\wt_i\lambda(R_i)-1+\wt_i/\wt_{\max}$. Taking the ceiling gives~\eqref{eq:wef1-wmms-integer}. Multiplying by $a_i$ and applying \cref{lem:wmms-formula} gives~\eqref{eq:wef1-wmms-additive}.

It remains to prove that the bounds cannot be improved. Fix $i$ with $\wt_i<\wt_{\max}$, and choose $h$ with $\wt_h=\wt_{\max}$. Fix a common magnitude $a>0$ and give every agent the same values on three items: $\val(p)=a$ and $\val(c_1)=\val(c_2)=-a$.
Equivalently, the normalized values are $u(p)=1$ and $u(c_1)=u(c_2)=-1$.
Allocate $A_i=\{c_1\}$, $A_h=\{p,c_2\}$, and the empty bundle to every other agent. The normalized bundle values are $-1,0,\dots,0$. Agent $i$ can remove $c_1$, while every other agent has no weighted envy. Thus the allocation is WEF1.

Here $R_i=-1$ and $\lambda(-1)=-1/\wt_{\max}$. At this value, every maximum-entitlement agent contributes $-1$ to the sum in~\eqref{eq:lambda-def}, while every smaller-entitlement agent contributes zero; any larger value of $\lambda$ makes every ceiling nonnegative. Therefore, $\WMMS_i=-a \cdot \frac{\wt_i}{\wt_{\max}}$, $\val_i(A_i)=-a$,
so equality holds in~\eqref{eq:wef1-wmms-additive}. Moreover, the right-hand side of~\eqref{eq:wef1-wmms-integer} equals $-1$, so equality also holds there.
\end{proof}

The theorem has three immediate consequences. Maximum-entitlement agents receive exact WMMS, all agents receive exact MMS when entitlements are equal, and every agent can reach WMMS after one item addition or removal of the type allowed below.

\begin{corollary}\label{cor:wef1-wmms-one-item}
Let $A$ be a WEF1 allocation of an equal-magnitude mixed instance.
\begin{enumerate}[(i)]
    \item If $\wt_i=\wt_{\max}$, then $\val_i(A_i)\ge\WMMS_i$.
    \item If all entitlements are equal, then $A$ is MMS-fair.
    \item For every agent $i$, either $\val_i(A_i)\ge\WMMS_i$, or there is an item $g\in\M\setminus A_i$ with $\val_i(g)=a_i$ such that $\val_i(A_i\cup\{g\})\ge\WMMS_i$, or there is an item $c\in A_i$ with $\val_i(c)=-a_i$ such that $\val_i(A_i\setminus\{c\})\ge\WMMS_i$.
\end{enumerate}
\end{corollary}

\begin{proof}
Parts (i) and (ii) follow directly from \cref{thm:wef1-wmms-mixed}. For part (iii), let $t_i:=\left\lceil\wt_i\lambda(R_i)\right\rceil$.
Since $\wt_i/\wt_{\max}>0$, \cref{eq:wef1-wmms-integer} implies $u_i(A_i)\ge t_i-1$. If $u_i(A_i)\ge t_i$, then \cref{lem:wmms-formula} gives exact WMMS. Otherwise, $u_i(A_i)=t_i-1$.

Suppose that no item outside $A_i$ has normalized value $1$ and no item in $A_i$ has normalized value $-1$. Then all items that $i$ values positively belong to $A_i$, and $A_i$ contains no item that she values negatively. Hence
\begin{equation}\label{eq:one-item-net-value}
    0\le u_i(A_i)
    \quad\text{and}\quad
    R_i\le u_i(A_i).
\end{equation}
For every feasible value $\lambda$ in~\eqref{eq:lambda-def},
\[
    \lambda
    =
    \sum_{j\in\N}\lambda\wt_j
    \le
    \sum_{j\in\N}\left\lceil\lambda\wt_j\right\rceil
    \le
    R_i.
\]
Thus $\lambda(R_i)\le R_i$. If $\lambda(R_i)\le0$, then $t_i\le0\le u_i(A_i)$. If $\lambda(R_i)>0$, then
\[
    \wt_i\lambda(R_i)
    \le
    \lambda(R_i)
    \le
    R_i
    \le
    u_i(A_i),
\]
and the integrality of $u_i(A_i)$ gives $t_i\le u_i(A_i)$. Both cases contradict $u_i(A_i)=t_i-1$. Therefore, an item of one of the two stated types exists. Adding or removing it raises normalized utility by one, from $t_i-1$ to $t_i$, and \cref{lem:wmms-formula} completes the proof.
\end{proof}

Additionally, bounded entitlement ratios give a uniform version of the additive guarantee.

\begin{corollary}\label{cor:wef1-wmms-ratio}
Let $\wt_{\min}:=\min_{j\in\N}\wt_j$ and $\kappa\ge1$. If $\wt_{\max}/\wt_{\min}\le\kappa$, then every WEF1 allocation of an equal-magnitude mixed instance satisfies
\[
    \val_i(A_i)
    \ge
    \WMMS_i-a_i\left(1-\frac{1}{\kappa}\right)
    \quad\text{for every }i.
\]
For every $\kappa>1$, no smaller coefficient than $1-1/\kappa$ holds for all such instances, even with two agents and identical valuations over one common good and two common chores.
\end{corollary}

Indeed, the ratio condition gives $\wt_i/\wt_{\max}\ge1/\kappa$ for every $i$, so the guarantee follows immediately from \cref{eq:wef1-wmms-additive}. The coefficient is best possible: for any $\kappa>1$, take two agents with entitlements $1/(1+\kappa)$ and $\kappa/(1+\kappa)$ and apply the three-item construction in the proof of \cref{thm:wef1-wmms-mixed} to the smaller-entitlement agent.

The tight example in the proof of \cref{thm:wef1-wmms-mixed} has negative WMMS. This is not essential: equality in \cref{eq:wef1-wmms-additive} can also hold when the smaller-entitlement agent's WMMS is positive. Let $k\ge2$ be an integer, take two agents with entitlements $1/(k+1)$ and $k/(k+1)$, fix a common magnitude $a>0$, and give them identical normalized values over $k+1$ common goods and one common chore. Give the first agent one good and the chore, and give the second agent the remaining $k$ goods. The allocation is WEF1: after the first agent removes the chore, both sides of her weighted comparison equal $k+1$, while the second agent has no weighted envy. Her utility is zero, while $R_1=k$ and a direct calculation from~\eqref{eq:lambda-def} gives us
\[
    \lambda(k)=\frac{(k+1)(k-1)}{k},
    \quad
    \WMMS_1=a\frac{k-1}{k}>0.
\]
At the displayed value of $\lambda$, the two ceilings in~\eqref{eq:lambda-def} are $1$ and $k-1$; any increase raises the second ceiling. Thus her loss is $a(1-1/k)=a(1-\wt_1/\wt_2)$, with equality in~\eqref{eq:wef1-wmms-additive}.

Finally, we show that the equal-magnitude assumption cannot simply be relaxed by allowing one larger positive value.

\begin{proposition}\label{thm:wmms-three-item}
For every real number $r>1$, there is a two-agent, three-item mixed instance with singleton values in $\{-1,1,r\}$ that has no exact WMMS allocation.
\end{proposition}

\begin{proof}
Let $\wt_1=\frac{1}{2r+1}$ and $\wt_2=\frac{2r}{2r+1}$.
There are three items $x,y,z$ with values as follows.
\[
\begin{array}{c|ccc}
 &x&y&z\\ \hline
\val_1&-1&1&r\\
\val_2&1&r&r
\end{array}
\]
For agent $2$, consider the ordered partition $P_1=\{x\}$ and $P_2=\{y,z\}$.
Both normalized bundle values are $2r+1$: $\frac{\val_2(P_1)}{\wt_1}=2r+1$ and $\frac{\val_2(P_2)}{\wt_2}=2r+1$.

For any ordered partition, the minimum normalized bundle value is at most its weighted average,
\[
    \sum_{j=1}^2
    \wt_j\frac{\val_2(P_j)}{\wt_j}
    =
    \val_2(\M)
    =
    2r+1.
\]
Therefore, $\WMMS_2=\wt_2(2r+1)=2r$.

Agent $1$ has strictly positive WMMS. Indeed, the ordered partition $P_1=\{x,z\}$, $P_2=\{y\}$ has bundle values $r-1>0$ and $1>0$.

To give agent $2$ utility at least $2r$, an allocation must give her $\{y,z\}$ or all three items. Every other proper subset is worth at most $r+1<2r$. In the first case, agent $1$ receives $\{x\}$, worth $-1$; in the second case, agent $1$ receives the empty bundle, worth $0$. Both values are below her strictly positive WMMS. Hence no exact WMMS allocation exists.
\end{proof}

The example uses only two agents and three items, has no zero values, and works for every $r>1$. Thus the exact theorem does not extend to the seemingly nearby value set $\{-1,1,r\}$.

\subsection{No Multiplicative Guarantee under Unrestricted Entitlements}\label{subsec:compatibility}

The preceding results use an additive loss measured in the agent's singleton-value magnitude. This form remains meaningful when a mixed manna WMMS value is zero or negative, cases in which a multiplicative ratio is not informative~\citep{KulkarniMehtaTaki20}. Even in the chore subclass, where nonnegative costs permit the standard multiplicative definition, no fixed factor can be combined with WEF1 when entitlements may be arbitrarily unequal. The following two-chore example makes this distinction precise. It does not contradict \cref{thm:wef1-wmms-mixed}: the additive loss for the smaller-entitlement agent remains below the cost of one item.

For instances with only chores, multiplicative WMMS guarantees are stated using nonnegative costs. Let $c_i=-\val_i$. Define agent $i$'s weighted maximin cost by
\begin{equation}\label{eq:wmmc-def}
    \WMMC_i
    :=
    \wt_i
    \min_{P\in\Pi_n(\M)}
    \max_{j\in\N}
    \frac{c_i(P_j)}{\wt_j}.
\end{equation}
An allocation is a $\rho$-WMMS allocation for chores if $c_i(A_i)\le \rho\,\WMMC_i$ for every $i$, where $\rho\ge1$.

\begin{theorem}\label{thm:wef1-wmms}
For every finite $\rho\ge1$, there is a two-agent instance with two identical unit-cost chores for which no allocation is simultaneously WEF1 and $\rho$-WMMS.
\end{theorem}

\begin{proof}
Let $\wt_1=\varepsilon$ and $\wt_2=1-\varepsilon$,
where $0<\varepsilon<\frac{1}{2\rho+1}$.
Both chores have cost $1$ for both agents.

Since $\varepsilon<1/3$, the minimum in~\eqref{eq:wmmc-def} assigns both chores to the bundle labeled for agent $2$. The maximum normalized cost is then $\frac{2}{1-\varepsilon}$.

A split partition has maximum normalized cost $1/\varepsilon$, which is larger. Thus, we have that
\begin{equation*}
    \WMMC_1=\frac{2\varepsilon}{1-\varepsilon},
    \quad
    \WMMC_2=2.
\end{equation*}

A WEF1 allocation cannot assign both chores to one agent. If it did, then after deleting one chore from her own bundle, that agent would still have positive normalized cost while the other bundle has cost zero. Therefore every WEF1 allocation splits the chores, giving agent $1$ cost $1$.

For a $\rho$-WMMS guarantee, agent $1$ would need $1
    \le
    \rho\frac{2\varepsilon}{1-\varepsilon}$.
However, the choice of $\varepsilon$ gives $\rho\frac{2\varepsilon}{1-\varepsilon}<1$.
Thus no WEF1 allocation is $\rho$-WMMS.
\end{proof}

Equivalently, the ratio between agent $1$'s cost in any WEF1 allocation and her weighted maximin cost is
\[
    \frac{1-\varepsilon}{2\varepsilon} \rightarrow \infty \quad \text{as } \varepsilon \rightarrow 0.
\]
Thus the additive comparison in \cref{thm:wef1-wmms-mixed} does not convert into a fixed multiplicative factor when the smallest entitlement approaches zero.

\section{Conclusion}\label{sec:conclusion}

This paper resolves the open existence question for WEF1 in mixed manna and proves that WMMS allocations exist in a structured valuation setting. Our results show that the two fairness viewpoints have different strengths. WEF1 has a general existence theorem but no utilitarian approximation guarantee. Exact WMMS requires structure in the values, but under equal-magnitude mixed values it can be computed in polynomial time together with fractional Pareto optimality.

We show that the ratio between unconstrained optimal welfare and the best WEF1 welfare is unbounded, already for two unweighted agents who agree on the sign of every item and have normalized total values. The lower bound remains valid when the chosen WEF1 allocation is fractionally Pareto optimal. Thus Pareto efficiency and utilitarian approximation answer different questions in mixed manna.

For WMMS, the equal-magnitude assumption makes the partition benchmark exactly tractable. After normalization, an agent's WMMS depends only on her total number of liked items minus disliked items. The ceiling formula turns this one-agent partition problem into an integer target, and the same inequalities prove that all agents' targets can be met together: a flow assigns the required liked items, and the remaining liked items are enough to distribute the common chores. Assigning every item to an agent with maximum normalized value yields fractional Pareto optimality at the same time.

The same formula gives a direct relation between the two fairness viewpoints. Every WEF1 allocation gives agent $i$ utility at least $\WMMS_i-a_i (1-\frac{w_i}{w_{\max}})$.
The bound is best possible, gives exact WMMS to every maximum-entitlement agent, and becomes exact MMS under equal entitlements. It also has a simple one-item interpretation: an agent below WMMS can reach it by adding one liked outside item or removing one disliked item from her own bundle. The three-item and two-chore examples mark the limits of this conclusion: a second positive magnitude can destroy exact WMMS existence, and unrestricted entitlement ratios rule out any fixed multiplicative comparison compatible with WEF1.

Our work leaves open two natural directions. First, does every general mixed manna instance admit a WEF1 and PO allocation? This remains open in the setting with equal entitlements~\citep{AzizEtAl26BoBW}. The stronger requirement of fPO cannot always be combined with WEF1~\citep[Theorem~4.3]{MackenzieSuzuki26}. Our infinite utilitarian price addresses a different question: it shows that even a WEF1 and fPO allocation may have welfare arbitrarily far from the maximum. Second, it would be useful to identify broader mixed valuation classes that admit exact WMMS under arbitrary entitlements, or in which WEF1 implies a meaningful additive WMMS guarantee. The three-item counterexample shows that such extensions need assumptions beyond merely allowing a small number of singleton values.

\section*{Acknowledgements}
The author used GPT-5.6 Sol as a research assistant when exploring proof ideas and revising the exposition. He wrote the mathematical arguments and takes full responsibility for the paper.

\bibliographystyle{plainnat}
\bibliography{bib}

\clearpage 

\appendix

\section{General Two-Sided RWPS Comparisons and a WEF1 Condition}\label{sec:multi-reserve}

The continuous representation of weighted picking sequences used below follows the analyses of \citet{LiLiWu22,WZZ25}. The difference from the usual RWPS
comparison is that the last $a$ actual picks of the observer are deleted while the first $b$ actual picks of the recipient are also set aside. We determine exactly when this comparison holds for every cost vector and then apply it to allocations in which one recipient may receive several acceptable bundles that an observer values nonnegatively.

The multi-chore branch of \cref{alg:wef1} uses pairwise disjoint acceptance sets. Consequently, for any observer $i$ and recipient $j$, agent $i$ values at most one acceptable bundle assigned to $j$ nonnegatively. The proof therefore needs to set aside at most one of $j$'s first actual picks when applying \cref{lem:two-sided}.

This appendix gives the corresponding RWPS comparison when several acceptable bundles assigned to $j$ are nonnegative for observer $i$. Each such bundle is paired with one of $j$'s first actual picks. Under the same negativity condition as in \cref{lem:bundling}(iii), every bundle-chore pair is negative for $i$, so these bundles do not increase $i$'s value for $j$'s final bundle.

The comparison is most easily understood through normalized pick counts. Removing the last $a$ actual picks of $i$ leaves $(k_i-a)/\wt_i$ picks per unit entitlement. Setting aside the first $b$ actual picks of $j$ leaves $(k_j-b)/\wt_j$. If the first quantity is no larger than the second, the cost comparison obtained by integrating the RWPS cost functions continues to hold. Constant costs show that this condition is exact for a statement that must hold for every valuation.

For integers $1\le a\le k_i$ and $0\le b\le k_j$, define
\begin{align*}
    L_i(a)&:=\{e_{i,1},\dots,e_{i,a}\},
    &&\text{the last $a$ actual picks of $i$},\\
    F_j(b)&:=\{e_{j,k_j-b+1},\dots,e_{j,k_j}\},
    &&\text{the first $b$ actual picks of $j$},
\end{align*}
where $F_j(0):=\varnothing$. Recall that actual picks occur in the reverse of the forward schedule: $L_i(a)$ corresponds to $i$'s earliest $a$ forward occurrences, whereas $F_j(b)$ corresponds to $j$'s latest $b$ forward occurrences.

\begin{theorem}[RWPS comparison after removing several picks]
\label{thm:several-rwps}
For every observer $i$ and recipient $j$, RWPS guarantees
\begin{equation}\label{eq:multi-rwps}
    \frac{d_i(C_i\setminus L_i(a))}{\wt_i}
    \le
    \frac{d_i(C_j\setminus F_j(b))}{\wt_j}
\end{equation}
whenever
\begin{equation}\label{eq:multi-count}
    \frac{k_i-a}{\wt_i}
    \le
    \frac{k_j-b}{\wt_j}.
\end{equation}
The count condition cannot be weakened for a guarantee that holds for every valuation: if all chores have the same cost to observer $i$, then~\eqref{eq:multi-rwps} holds if and only if~\eqref{eq:multi-count} holds.
\end{theorem}

\begin{proof}
We follow the proof of \cref{lem:two-sided}, shifting the cost comparison by $a/\wt_i$ instead of $1/\wt_i$. Fix
$\alpha\in(a/\wt_i,k_i/\wt_i]$ that is not an interval endpoint, and suppose $\rho_i(\alpha)=d_i(e_{i,z})$. Then $z>a$. At the beginning $t^*$ of $i$'s $z$th forward occurrence,
\[
    \frac{z-1}{\wt_i}=s_i(t^*)\le s_j(t^*).
\]
Moreover,
\[
    \alpha-\frac{a}{\wt_i}
    \le
    \frac{z-a}{\wt_i}
    \le
    \frac{z-1}{\wt_i}.
\]
Thus $j$ reaches normalized count $\alpha-a/\wt_i$ no later than $t^*$. The corresponding forward occurrence of $j$ is therefore no later than the occurrence associated with $e_{i,z}$. Since the chores are chosen in reverse schedule order, that chore of $j$ is chosen no earlier than $e_{i,z}$ and is still available when $i$ chooses $e_{i,z}$, except at interval endpoints. Agent $i$ chooses a remaining chore of minimum $d_i$-cost, so
\[
    \rho_i(\alpha)
    \le
    \rho_j^i\left(\alpha-\frac{a}{\wt_i}\right)
    \quad\text{for almost every }\alpha.
\]
Integrating and changing variables gives us
\begin{align*}
    \frac{d_i(C_i\setminus L_i(a))}{\wt_i}
    &=\int_{a/\wt_i}^{k_i/\wt_i}\rho_i(\alpha)\,d\alpha \le
    \int_0^{(k_i-a)/\wt_i}\rho_j^i(\beta)\,d\beta \le
    \int_0^{(k_j-b)/\wt_j}\rho_j^i(\beta)\,d\beta =\frac{d_i(C_j\setminus F_j(b))}{\wt_j},
\end{align*}
where the second inequality uses~\eqref{eq:multi-count} and the nonnegativity of the cost function.

To see that the count condition is exact, let every chore have cost $1$ to observer $i$. The two sides of~\eqref{eq:multi-rwps} are then $(k_i-a)/\wt_i$ and $(k_j-b)/\wt_j$, respectively. Hence~\eqref{eq:multi-rwps} holds exactly when~\eqref{eq:multi-count} holds.
\end{proof}

\Cref{thm:several-rwps} contains both comparisons in \cref{lem:two-sided}: the choices $(a,b)=(1,0)$ and $(a,b)=(1,1)$ give its first and second inequalities, respectively.

\begin{corollary}[Number of chores that may be set aside]\label{cor:reserved-count}
After deleting the last $a$ actual picks of $i$, the maximum number of $j$'s first actual picks that can always be set aside is
\begin{equation}\label{eq:reserved-count}
    b^{\max}_{ij}(a)
    =k_j-
    \left\lceil\frac{\wt_j(k_i-a)}{\wt_i}\right\rceil.
\end{equation}
In particular, $b^{\max}_{ij}(1)\ge1 \iff q_i\le q_j$.
\end{corollary}

\begin{proof}
By the load bound in~\eqref{eq:load-bound} and $a\ge1$, the choice $b=0$ satisfies~\eqref{eq:multi-count}. The largest integer $b$ satisfying that condition is
\[
    \left\lfloor k_j-\frac{\wt_j(k_i-a)}{\wt_i}\right\rfloor
    =k_j-\left\lceil\frac{\wt_j(k_i-a)}{\wt_i}\right\rceil.
\]
For $a=1$, the inequality $b^{\max}_{ij}(1)\ge1$ is equivalent to
$(k_i-1)/\wt_i\le(k_j-1)/\wt_j$, which is $q_i\le q_j$.
\end{proof}

The corollary gives a direct mixed manna interpretation. If observer $i$ values $r$ acceptable bundles assigned to recipient $j$ nonnegatively, then the proof can pair those bundles with $r$ of $j$'s first actual picks whenever $r\le b^{\max}_{ij}(1)$. This removes the need for disjoint acceptance sets, provided the following count condition holds for every ordered pair.

\begin{theorem}\label{thm:multi-bundle}
Assume $|\Chores|\ge n$ and allocate the objective chores by RWPS. Let $\Bundles$ be a pairwise disjoint family of bundles, disjoint from $\Chores$, such that $\Chores\cup\bigcup_{B\in\Bundles}B=\M$.
Suppose these bundles satisfy the negativity condition
\begin{equation}\label{eq:general-negativity}
    \val_i(B\cup\{c\})<0
    \quad
    \text{for all }i\in\N,\ B\in\Bundles,\ c\in\Chores.
\end{equation}
Assign every $B\in\Bundles$ to an agent $a(B)$ with $\val_{a(B)}(B)\ge0$, and define $A_i:=C_i\cup\bigcup_{B:\,a(B)=i}B$.
For $i\ne j$, let $r_{ij}:=
    |
    \{B\in\Bundles:a(B)=j\text{ and }\val_i(B)\ge0\}
    |$.
If, for every ordered pair $i\ne j$,
\begin{equation}\label{eq:multi-bundle-condition}
    \frac{k_i-1}{\wt_i}
    \le
    \frac{k_j-r_{ij}}{\wt_j},
\end{equation}
then the resulting allocation is WEF1. More strongly, deleting $x_i$ eliminates all weighted envy of agent $i$.
\end{theorem}

\begin{proof}
Fix $i\ne j$ and write $r=r_{ij}$. Since $k_i\ge1$, the left-hand side of~\eqref{eq:multi-bundle-condition} is nonnegative, so the condition also implies $r\le k_j$. We may therefore apply \cref{thm:several-rwps} with $a=1$ and $b=r$:
\begin{equation}\label{eq:multi-bundle-rwps}
    \frac{\val_i(C_i\setminus\{x_i\})}{\wt_i}
    \ge
    \frac{\val_i(C_j\setminus F_j(r))}{\wt_j}.
\end{equation}
Every bundle assigned to $i$ is nonnegative for $i$, so adding these bundles to the left-hand side does not hurt the inequality.

Among the bundles assigned to $j$, exactly $r$ are nonnegative for observer $i$. Pair them bijectively with the $r$ chores in $F_j(r)$. By~\eqref{eq:general-negativity}, every bundle-chore pair has negative value to $i$. Every other bundle assigned to $j$ has negative value to $i$. Therefore
\[
    \val_i(A_j)
    \le
    \val_i(C_j\setminus F_j(r)).
\]
Combining this inequality with~\eqref{eq:multi-bundle-rwps} proves the WEF1 comparison after deleting $x_i$. Since $x_i$ is an objective chore, $\val_i(x_i)<0$, so this deletion is allowed by \cref{def:wef1}.
\end{proof}

The preprocessing in \cref{sec:bundling} makes the acceptance sets pairwise disjoint, so $r_{ij}\le1$. If $r_{ij}=0$, the standard RWPS load bound gives
$q_i\le k_j/\wt_j$, which is~\eqref{eq:multi-bundle-condition}. If $r_{ij}=1$, the bundle assigned to $j$ is also acceptable to $i$, and choosing a recipient of maximum $q$ gives $q_i\le q_j=(k_j-1)/\wt_j$. Hence the multi-chore branch of \cref{alg:wef1} is exactly the case of \cref{thm:multi-bundle} in which each observer values at most one bundle assigned to a recipient nonnegatively. Pairwise disjoint acceptance sets are therefore a simple sufficient condition for the general count requirement.

\end{document}